\documentclass[a4paper]{article}

\usepackage{arxiv}

\pdfoutput=1

\usepackage{amssymb,amsmath}
\usepackage{xspace}
\usepackage[ruled,vlined]{algorithm2e}

\usepackage{amsthm}

\newcommand{\setObjects}{\mathcal{O}\xspace}
\newcommand{\setVoters}{\mathcal{\mathcal{V}}\xspace}

\newcommand{\object}{o\xspace}

\newcommand{\voter}{v\xspace}

 \newcommand{\weight}{w\xspace}
 \newcommand{\declaration}{b\xspace}
 \newcommand{\Declaration}{\mathbf{B}\xspace}
\newcommand{\CM}{$(\setVoters,\setObjects,(\utility_{v})_{v \in \setVoters} ,Q,\Si,\SCF)$}
\newcommand{\SCM}{$(\setVoters,\setObjects,W,(\declaration^*_{v})_{v \in \setVoters},\SCF)$}
\newcommand{\SCF}{\mathcal{A}lgo}
\newcommand{\SCMU}{$(\setVoters,\setObjects,W,(\declaration^*_{v})_{v \in \setVoters},\SCF_{utilitarian})$}
\newcommand{\SCMF}{$(\setVoters,\setObjects,W,(\declaration^*_{v})_{v \in \setVoters},\SCF_{fair})$}

\newcommand{\utility}{{u}\xspace}

\newcommand{\Si}{\mathcal{S}} 
\newcommand{\motnouveau}[1]{\emph{#1}}

\usepackage{tikzsymbols}

\theoremstyle{plain}
\newtheorem{theorem}{Theorem}
\newtheorem{lemma}[theorem]{Lemma}

\newtheorem{proposition}[theorem]{Proposition}

\newtheorem{definition}[theorem]{Definition}

\theoremstyle{definition}

\theoremstyle{remark}

\newtheorem*{note*}{Note}

\newtheorem*{remark*}{Remark}
\theoremstyle{claimstyle}

\newtheorem*{claim*}{Claim}

\newtheorem*{lemma*}{Lemma}
\newtheorem*{proposition*}{Proposition}
\newtheorem*{definition*}{Definition}
\newtheorem*{theorem*}{Theorem}

\usepackage[pdfa,unicode]{hyperref}%
\usepackage{cleveref}

\begin{document}

\title{Designing Strategyproof Election Systems with Score Voting}

\author{Johanne Cohen\\
        LISN-CNRS, Universit\'e Paris-Saclay, France\\
        \texttt{johanne.cohen@lri.fr}\\
        \And
        Daniel Cordeiro\\
        University of S\~ao Paulo, Brazil\\
        \texttt{daniel.cordeiro@usp.br}\\
        \And
        Valentin Dardilhac\\
        Universit\'e Paris-Saclay, France\\
        valentin.dardilhac@ens-paris-saclay.fr\\
        \And
        Victor Glaser\\
        \'Ecole normale sup\'erieure de Lyon, France\\
        \texttt{victor.glaser@ens-lyon.fr}}

\maketitle

\begin{abstract}
We focus on the strategyproofness of voting systems where voters must choose a number of options among several possibilities. These systems include those that are used for Participatory Budgeting, where we organize an election to determine the allocation of a community's budget (city, region, etc.) dedicated to the financing of projects.

We present a model for studying voting mechanisms and the Constrained Change Property (CCP), which will be used to design voting mechanisms that are always strategyproof. We also define a new notion of social choice function and use it to design a new class of utilitarian voting mechanisms that we call score voting. We prove that the mechanisms designed with core voting with a neutral score function are equivalent to knapsack voting on the same instance and that any score voting designed with a total score function is strategyproof if and only if its score function satisfies CCP. 

These results are combined to devise an algorithm that can find the  closest total score function that makes any given score voting to be strategyproof.
\end{abstract}


\newpage
\section{Introduction} \label{Section:1}

Social choice theory is a branch of science that studies how individual preferences can be aggregated in a collective choice~\cite{Bra16}. Social choice theory has been applied to study applications in several domains. In particular, it has been used to study what is known in the literature as \emph{Knapsack voting}, adapted from Cabannes' idea of Participatory Budgeting~\cite{cabannes2004participatory}.

Knapsack Voting (KP) seeks to invite citizens to participate in the process of deciding how public money is spent. This form of participatory democracy was first employed by the city of Porto~Alegre, Brazil, in 1989. Since then it has been used by different cities around the world like Madrid, Seoul, Bogota, New York, and Paris. For instance, in 2016 Paris applied KP to allow citizens to vote on how to allocate a budget of 100 million Euros~\cite{cabannes2020another}.

Citizens vote independently for a subset of projects considering multiple criteria (cost of the project, location, beneficiaries, etc.). Their choices are then used to reach a joint decision in a fair and principled way. 

In this work, we are interested in studying the properties of voting mechanisms---algorithms that select a solution by taking into account the opinion of the voters---that make voting resilient to manipulation. We study the concept of \emph{strategyproofness} (the equivalent for auctions being called truthfulness) of those mechanisms, which is the idea that the best voting strategy for a voter is to be sincere, i.e., the player has no incentive to strategically change its preferred vote in order to increase its outcome.

Designing a strategyproof voting mechanism is hard due to several impossibilities results resulting from the Gibbard-Satterthwaite theorem \cite{Gib73, sat75} and its extensions. The main contributions of this work are the following.

The first contribution is a new model for voting mechanisms that allows the study of voting mechanisms regardless of their type (utilitarian or fair). We called this model the \emph{common Choice Mechanism} (CM). Using this model, we describe conditions that make non-dictatorial mechanisms non-strategyproof and show how the social choice functions respecting the ``Constrained Change Property'' (CCP) can be used to design CMs that are strategyproof for the unitary case.

The second contribution of this work is the notion of \emph{score functions}. We use this notion to design a new class of utilitarian voting mechanisms that we call \emph{score voting}. We show that the mechanisms designed with score voting that have the ``neutrality property'' are equivalent to knapsack voting. We also show that any score voting designed with a \emph{total} score function is strategyproof if and only if its score function satisfies CCP. We present an algorithm that use this result to find the closest total score function that makes a score voting strategyproof.

The remaining of this document is organized as follows. \Cref{sec:related} presents works on computational social choice and Participatory Budgeting. In \cref{sec:notations} we present a model to study choice mechanisms (CMs) and the notations used in this document; we also formally define the notion of strategyproofness. Section~\ref{sec:impossibility} studies the strategyproofness of different mechanisms and shows how the CCP property can be used to design CMs that are strategyproof. \Cref{sec:scorevoting} studies several \emph{score voting} mechanisms, shows how they can be used to design strategyproof mechanisms, and present an algorithm to compute the closest strategyproof total score function. Finally, the appendix details all the omitted proofs.

\section{Related Work}
\label{sec:related}

Participatory budgeting (PB) \cite{cabannes2004participatory,cabannes2020another} has been applied by different municipalities as a democratic tool to allow citizens to prioritize investments on several projects given a limited budget. The idea was first applied in 1989 in the city of Porto Alegre, Brazil, and has been used in several cities, notably in Latin America and Europe~\cite{Aziz2021}.

Different electoral systems and their properties have been studied by \emph{Computational social choice}~\cite{Bra16}, a branch of science that studies the computational aspects of collective decision-making. It covers problems regarding voting theory (including mechanisms design, the computational complexity of choosing a winner, strategic voting), fairness in allocations, coalition formation, etc.

In particular, there is a great deal of interest in studying the manipulation of decisions by decision-makers. Voters can strategically change their true preferences in order to obtain a better outcome. Decision-making mechanisms that are immune to strategic voting are called \emph{strategyproof} (also called truthful depending on the domain).

A common electoral system used in participatory budgeting is the $k$-approval voting. Each voter chooses (``approves'') up to $k$ projects and the projects with the highest number of approvals are funded (respecting the budget constraints).

Goel et al.~\cite{Goel19} introduced the \motnouveau{Knapsack Voting} scheme. The idea comes from the fact that applying PB is conceptually similar to solving the classical Knapsack problem, with the set of chosen budget items fitting a limited budget $B$ while maximizing societal value \cite{chen2016mathematical}. In this scheme, each citizen votes for a subset of the objects such that the sum of the costs of the objects satisfies the budget constraint. They showed that this schema is strategyproof and welfare-maximizing when the outcome for the voter is given by the $\ell^1$ distance from the outcome and its true preference and partially strategyproof under additive concave utilities. 

Aggregating budget division schemes that maximize the utilitarian social welfare of voters have a tendency to overprioritize majority preferences, resulting in unfairness problems~\cite{Fre21}. The seminal work by Moulin~\cite{moulin1980strategy} shows that $\ell^1$ preferences are a particular case of \textit{single-peaked} preferences and presents a family of voting schemes that are both incentive compatible and proportional by adding some fixed (``phantom'') ballots to the voter's ballots and choosing the median of the larger set. This result was later generalized by several works~\cite{barbera1994characterization,border1983straightforward,duddy2015fair,Fre21,peters1992pareto}.

Voting schemes that maximize the utilitarian social welfare are possible because \textit{single-peaked} preferences assume that there exists an ordering on the alternatives. More general mechanisms may not be strategyproof due to
an important impossibility result was independently proved by 
Gibbard~\cite{Gib73} and Satterthwaite~\cite{sat75}. The Gibbard-Satterthwaite theorem states that every resolute, non-imposed, and non-dictatorial social choice function for three or more alternatives is susceptible to strategic manipulation \cite{Bra16}.

\section{Preliminaries for collective decision-making problems }
\label{Section:2}
\label{sec:notations}

The notations used in this paper are based on the standard notation given by Brandt et al.~\cite{Bra16}.

A \motnouveau{common Choice Mechanism} (CM) is a $6$-uplet \CM\ which is defined as follows.  We consider a set of voters $\setVoters = [n]$ (where $[i] = \{1,\dots, i\}$, for $i \in \mathbb{N}$)
and a finite set of $m$ alternatives (objects) $\setObjects$ that answers the set of questions Q requested from the voters. The set of all valid possible collective decisions is denoted by $\Si \subseteq \setObjects$.

Each voter $\voter \in \setVoters$ has a private preference over all solutions given by a utility function\footnote{The set of all subsets of $S$ is  denoted by $\mathcal{P}(S)$.} $\utility_\voter: \mathcal{P}(\setObjects) \to \mathbb{R}$. According to a set of questions \emph{$Q$}, that defines the format of the ballots, each voter $\voter$ answers these questions via a ballot  $\declaration_{\voter}$.
A  \motnouveau{declaration profile} $\Declaration = (\declaration_{1}, \dots, \declaration_{n})$ consists of a ballot $\declaration_{\voter}$ for each voter $\voter$.
In the remaining of the document, $(\declaration_{\voter};\declaration_{-\voter})$ is a shorthand for $(\declaration_{1}, \dots,\declaration_\voter, \dots, \declaration_{n})$, used here to highlight the ballot of voter $\voter$ against that of all other voters.

Voting can mean different things depending on the specified form of a ballot and a collective decision.  For example, the ballot can be a subset $e\in \mathcal{P}(\setObjects)$,  a linear ordering of the objects, or the value function among all the objects. The \motnouveau{social choice function} $\SCF$, aggregates a set of individual ballots $\Declaration$ into a collective decision and returns a \motnouveau{solution}, which is a \motnouveau{winning set} of objects $\SCF(\Declaration)$.

Our study focuses on the design of the social choice function that incentives an individual voter to vote sincerely.
We denote by \emph{$ \declaration^*_\voter$} the ballot of the voter $\voter$ if she sincerely answers the questions based on her utility.  

\begin{definition}
We say that a CM is \motnouveau{strategyproof} if: for all profile ballots $\Declaration$, for all voters $\voter \in \setVoters$,  we have $
\utility_\voter (\SCF( \declaration^*_\voter;\Declaration_{- \voter}))\ge \utility_\voter (\SCF( \declaration_\voter ;\Declaration_{- \voter}))$.
\end{definition} 

Our work focuses on mechanisms designed to maximize social welfare assuming voters have utilities over alternatives. 

\begin{definition}
 If the ballot is a valuation of the objects, a \motnouveau{social welfare-maximizing function (SWF)} is an optimization function $\mathcal{P}(\setObjects)\to \mathbb{R} $ that maximizes a social welfare.
\end{definition}

In this work, we distinguish two forms of social welfare-maximizing functions:
\begin{itemize}
\item a  social welfare function denoted $f_{utilitarian}$  is \motnouveau{utilitarian}  if it maximizes $\sum_{\voter \in \setVoters}
\sum_{\object \in s}\declaration_\voter(\object)$, $\forall s \in \Si$;
\item  a social welfare function denoted $f_{fair}$, is  \motnouveau{fair} if maximizes  $\min_{\voter \in \setVoters}\sum_{\object \in s}\declaration_\voter(\object),$  $\forall s \in \Si$.
\end{itemize}

An algorithm $\SCF$ that finds the optimal solution of a utilitarian (resp. fair) function is denoted $\SCF_{utilitarian}$ (resp. $\SCF_{fair}$).

In order to study different types of ballots, we use a particular class of CMs called \motnouveau{Simple CMs} that takes into account different weights of objects and uses the sincere ballot of the voter in its utility. \motnouveau{Simple CM} is a $5$-uplet \SCM\ defined as follows. The ballots are functions: $\setObjects \to \mathbb{R}$. Every object $\object$ has also a weight denoted $\weight(\object)$ (hence, the weight is a function $\weight: \setObjects \to \mathbb{R}$). Value $W$ is a weight constraint that restricts valid solutions, so that $\Si=\{e \in \mathcal{P}(\setObjects): \sum_{o\in e} \weight(o) \leq W \}$.  The utility depends on the values of the  set in the sincere ballot: $\utility_\voter(e)=\sum_{\object \in e} \declaration^*_\voter(\object)$.

Now we define the type of the Simple CM according to the type of ballot:
   \begin{itemize}
       \item a simple CM is an \motnouveau{approval voting} if   the ballot $\declaration_{v}$ is a a function $\declaration_\voter: \setObjects\to \{0,1\}$. This ballot can also be defined as a vector of $\{0,1\}^{|\setObjects|}$;

     \item a simple CM is a \motnouveau{ranking voting} if the ballot $\declaration_{v}$  of voter $\voter$ is a rank of the objects, i.e.\ a bijection $\declaration_\voter: \setObjects\to \{1,|\setObjects|\}$;
     
       \item a simple CM is the  \motnouveau {value function voting} if the ballot is a valuation of the objects of $\setObjects$, i.e.\ a function $\declaration_\voter:\setObjects\to \mathbb{R}$.  
   \end{itemize}

 \textbf{Example of participatory budgeting problem:} Given a set of projects $\setObjects$ (the alternatives), a set of voters $\setVoters$ needs to select a set of projects they have identified to be interested in. Each project $\object$ has a cost $\weight(o)$, and there is a fixed total budget of $W$. The question $Q$ arises of which projects should be funded. Each voter $\voter$ votes for a subset $d_{v} \subseteq \setObjects$, such that it satisfies the budget constraint $W$ (i.e., 
$\sum_{\object \in d_{v} } \weight(\object) \leq W$).  Since the ballot is a subset $\setObjects$, it is also possible to describe participatory budgeting with a simple CM, which is \SCMU. The participatory budgeting problem has a utilitarian social welfare function because the objective is to find the set of objects that maximize the number of votes.


Finally, Knapsack voting is an approval voting \SCMU\  when all objects have the same weight and $\SCF_{utilitarian}$ returns the optimal solution for the utilitarian social welfare function.  

\section{Properties on different votes}
\label{sec:impossibility}

Before understanding the strategyproofness of CM, we focus on the optimization problem of the social welfare function. The considered optimization problems are formally related to the Knapsack problems. We focus on different types of ballots, different social welfare-maximizing functions (utilitarian or fair), and the weight of objects (unitary weight, i.e. all the objects have the same weight $1$, or without no restriction of weight). For simplicity, we gather all our results in Table~\ref{tab:complexityWeightOne} (all the proofs are in the appendix). 
Unfortunately, most of the variants of Knapsack problems are NP-complete (see \cite{Kellerer}). To our knowledge, no work in the literature studies the knapsack problem using the fair optimization function.

\begin{table}[h]

\centering

\begin{tabular}{ |c||c|c|c|  } 
 \hline
     Optimization     & \multicolumn{3}{c|}{Ballot}    \\\cline{2-4}
     objective &   approval voting & ranking voting &  value function    \\[0.2cm]\hline  \hline   
     utilitarian & \multicolumn{2}{c|}{P}  &  P   for the unitary case (Proposition \ref{proposition:p1} )\\
                    &  \multicolumn{2}{c|}{(Proposition \ref{proposition:p13} and Proposition \ref{proposition:p17})} & NP-complete  (Proposition \ref{proposition:p19})\\[0.2cm]\hline \hline 
    fair      &      \multicolumn{3}{c|}{ NP-complete even for the unitary case.} \\[0.2cm]\cline{2-4}   
              &       Proposition \ref{proposition:p7}    &  Proposition \ref{proposition:p9}  &  Proposition \ref{proposition:p11} \\[0.2cm]
 \hline
\end{tabular}

\caption{Complexity  to compute the winning set. }
\label{tab:complexityWeightOne}
\end{table}

\begin{table}[h]

\centering

\begin{tabular}{ |c||c|c|c|  } 
 \hline
     Optimization     & \multicolumn{3}{c|}{Ballot}    \\\cline{2-4}
     objective &   approval voting & ranking voting &  value function    \\[0.2cm]\hline \hline   
 &  strategyproof   for the unitary case & \multicolumn{2}{c|}{ Not-strategyproof even for the } \\
                  &     (Goel et al.~\cite{Goel19})      &    \multicolumn{2}{c|}{   unitary case}   \\[0.2cm]\cline{2-4} 
utilitarian                  &  not-strategyproof   for the  non-unitary case &  & \\
                  &     Proposition \ref{proposition:SP:AV:utilitarian:general}    &  Proposition \ref{proposition:SP:RV:utilitarian:general}  &  Proposition \ref{proposition:SP:VF:utilitarian:unitary} \\[0.2cm]\hline \hline 
     fair      &      \multicolumn{3}{c|}{ not-strategyproof even  for the unitary case} \\[0.2cm]\cline{2-4}   
                &     Proposition \ref{proposition:SP:AV:fair:unitary}    &  Proposition \ref{proposition:SP:RV:fair:general}  &  Proposition \ref{proposition:SP:VF:fair:unitary} \\[0.2cm]
 \hline
\end{tabular}

\caption{Strategyproofness on the CM when the social choice function returns the optimal solution of a social welfare-maximizing 
function.  }
\label{tab:Strategyproofness}
\end{table}


     


Now, we  study the strategyproofness and the complexity to compute the result on some simple CMs. We study scores and types of algorithms (utilitarian or fair). 

We study the strategyproofness property on simple CM when its social choice function returns the optimal solution of the social welfare-maximizing function. 
 The Gibbard-Satterthwaite theorem\footnote{Theorem~\ref{theo:2} is a rephrasing of this theorem using our notation.} gives some impossibility results about strategyproofness.

\begin{theorem}[\cite{Gib73}]\label{theo:2} 
  Whenever the utility $\utility_\voter$ is represented by a ranking of the objects, one of the following propositions is true:
\begin{itemize}
    \item The ballot only considers two possible outcomes (ex: a yes/no question);
    \item The social choice function  is dictatorial; a voter can choose the outcome;
    \item The CM is not strategyproof.
\end{itemize}
\end{theorem}

This theorem is  powerful because it can be used for many existing voting mechanisms. Goel \textit{et al.} shows a surprising result:

\begin{theorem}[\cite{Goel19}] \label{theo:1}
Unitary approval voting is strategyproof.
\end{theorem} 

We want to generalize Goel et al.'s result to other CM. Hence, we study the strategyproofness property on simple CM when its social choice function returns the optimal solution of the social welfare-maximizing function. Unfortunately, we obtain only results of impossibilities (by giving some counterexample or by applying the Gibbard-Satterthwaite theorem). For simplicity, we gather all our results in Table~\ref{tab:Strategyproofness}.


The Gibbard-Satterthwaite theorem cannot be easily extended to more general utilities. We use the common Choice Mechanism model to study the limitations of strategyproofness. Let us now describe the tools that will be used for an impossibility result on CM.

\begin{definition}
Let  CM=\CM\ be a common choice mechanism. A \emph{feasible} solution $s \in \Si$ is called \emph{first-class} set if for every 
$s^* \in \Si$, there is a permutation $\sigma_{s^*}$ such that $\forall ~ \voter \in \setVoters , \utility_{\sigma_{s^*}(v)}(s)\ge \utility_v(s^*)$.

A common choice mechanism CM=\CM\ is called
\motnouveau{first-class} if it always chooses a first-class set whenever there is one.
\end{definition}



The notion of a first-class CM can be used to show a new impossibility result:

\begin{theorem}\label{theo:3} 
A first-class CM ($E,\Si,\setVoters,\{\utility_\voter\}_{\voter\; \in \setVoters} ,Q,\mathcal{A}lgo$) is not strategyproof if both those conditions hold:
\begin{itemize}
    \item $Q$ forces the type of the ballots so that the ballot of voter~$\voter$ has to contain the feasible solution $m_\voter=argmax \{ \utility_\voter(s):s \in \Si \} $ and the feasible solution $m^*_\voter=argmax \{ \utility_\voter(s):s \in \Si \backslash \{m_\voter\}\}$. Plus, given the two preferred feasible solutions of a voter~$\voter$, $Q$ also forces the type of the ballots so that the ballots contain their images given by the utility function $u_\voter$, which must be distinct; and
    \item the utilities can take at least $3$ distinct values (e.g., $0,1,2$, or a ranking on at least $3$ objects, and only one is selected) independently for each feasible solution.
\end{itemize}
\end{theorem}

\begin{proof}
Lets take two voters $\voter_1$ and $\voter_2$ such that $m^*_{\voter_2}=m_{\voter_1}$ and $m^*_{\voter_1}=m_{\voter_2}$ and $\utility_{\voter_1}(m_{\voter_1})=\utility_{\voter_2}(m_{\voter_2})$ and $ \utility_{\voter_2}(m_{\voter_1})=\utility_{\voter_1}(m_{\voter_2})\neq \utility_{\voter_1}(m_{\voter_1})$.

The only first-class sets are $m_{v_1}$ and $m_{v_2}$ (because one of the voters of a first-class set must have a preference greater or equal than $\utility_{\voter_1}(m_{\voter_1})$ for this set, which is only the case for $m_{\voter_1}$ and $m_{\voter_2}$). Hence, if the ballots are sincere, the CM chooses either $m_{\voter_1}$ or $m_{\voter_2}$ (lets say w.l.o.g.\ that it is $m_{\voter_2}$). If $\voter_1$ lies on her ballot by telling that $\utility_{\voter_1}(m_{\voter_2})$ is lower than the reality (it does not change the other values of the utility because of their independence), then $m_{\voter_2}$ is not a first-class set anymore. Thus, $m_{v_1}$ is 
chosen (being the only first-class set), and the utility of $\voter_1$ increases. Hence, this CM is not strategyproof.
\end{proof}

Theorem~\ref{theo:3} applies in different cases than Gibbard-Satterthwaite's theorem because the next proposition is an application of Theorem~\ref{theo:3} (it is not the  case for Gibbard-Satterthwaite's theorem).

\begin{proposition} \label{proposition:p15}
There exists an approval voting \SCMF\  with $\utility_\voter(e)=\sum_{\object\in e\cap d^*_\voter}\weight(\object)$ such that the social choice function that returns the optimal solution is not strategyproof. 
\end{proposition}


Now, we characterize a sufficient property on CMs to be strategyproof. From now on, we only consider approval voting (the ballots can also be defined as a vector of $\{0,1\}^{|\setObjects|}$).


We find that another property seems essential to characterize strategyproofness.  From now on, any voting mechanism that we study will be used on at least $3$ objects. We focus on a compelling property:

\begin{definition}[Constrained Change Property (CCP)]
Let us consider the approval voting \SCM.
 The social choice function $\SCF$ respects the Constrained Change Property  (CCP)  if
 a voter $\voter$ switches from an object $\alpha$ to an object $\beta$ in a ballot, the social choice function can either:
\begin{itemize}
    \item Eject $\alpha$ of the output solution and have it replaced by another object;
    \item Get $\beta$ to be chosen and eject a previously chosen object (and only one);
    \item Not change the output solution.
\end{itemize}
\end{definition}

For the first condition, it is enough to consider an approval voting with a single voter. The social choice functions should select only the objects in her ballot: if the voter does not vote for the object $\alpha$, it must no longer be in the output solution. For the last two conditions, modifying a ballot by replacing an object with another does not have to alter the whole solution because the projects have the same cost. Hence there should be a relative symmetry between the objects. 

This property on the social choice function  implies the  strategyproofness (corresponding to the following theorem):

\begin{theorem}\label{theo:4} 
Let us consider an approval voting CM \SCM, in the unitary case.
If the social choice function respects CCP, then the CM is strategyproof.
\end{theorem}

\begin{proof}
%
Let $\voter$ be a voter and $\declaration_\voter^*$ her sincere ballot. Suppose that $\voter$ gives a ballot $\declaration_\voter$ that is not sincere, which is optimal to maximize  her utility.  We prove that voting $\declaration_\voter^*$ or $\declaration_\voter$ would give the same utility for voter $\voter$: \emph{i.~e},  $\utility_\voter(\SCF(\declaration_v;\Declaration_{-v}))=\utility_\voter(\SCF(\declaration_\voter^{*};\Declaration_{-v}))$  for all ballot profiles $\Declaration$.

We prove the statement by induction on a parameter $|\declaration_\voter\backslash \declaration_\voter^*|$, corresponding to the number of different objects between $\declaration_\voter$ and $\declaration_\voter^*$.

$|\declaration_\voter\backslash \declaration_\voter^*|=0$ means that $\declaration_\voter=\declaration_\voter^*$. So, the strategyproof ballot is optimal, and the statement for $k=0$ is proved. Suppose the statement holds when $|\declaration_\voter\backslash \declaration_\voter^*| \leq k$
for some value of $k\geq 0$ (Assumption $R_k$).

Now, assume  $|\declaration_\voter\backslash \declaration_\voter^*|=k+1$.
Moreover, we consider that the utility of $\voter$, having voted $\declaration_\voter$ is  than the one if she voted $\declaration_\voter^*$, because $\declaration_\voter$ is optimal (if they are the same, just set $\declaration_\voter=\declaration_\voter^*$, and the proof is over).

Lets take an object $\beta$ which is in $\declaration_\voter^*$ and not in $\declaration_\voter$ and $\alpha$ which is in $\declaration_\voter$ and not in $\declaration_\voter^*$. It is possible because  $\declaration_\voter^* \neq \declaration_\voter$.

Let $\declaration_\voter'$ be  $\declaration_\voter$ by adding $\beta$ and removing $\alpha$ : $\declaration_\voter'=\declaration_\voter-\{\alpha\} \cup \{\beta\} $. We  focus on $\utility_\voter(\SCF(\declaration_\voter';\Declaration_{-v}))$. Due to CCP property, there are three possibilities:\newline
\begin{enumerate}
    \item Assume that $\alpha$ gets ejected and an object $\gamma$ replaces it. In other words, we have $\SCF(\declaration_\voter';\Declaration_{-v}) = 
  \SCF(\declaration_\voter;\Declaration_{-v})-\{\alpha\}\cup \{\gamma\}$. We have  $\utility_\voter(\SCF(\declaration_\voter,\Declaration_{-v})\cup\{\gamma\}-\{\alpha\})\ge \utility_\voter(\SCF(\declaration_\voter,\Declaration_{-v})$, because 
$\declaration_\voter^*(\alpha)=0$ and $\declaration_\voter^*(\gamma)\ge 0$.  So, $\utility_\voter(\SCF(\declaration_\voter',\Declaration_{-\voter}))\ge \utility_\voter(\SCF(\declaration_\voter,\Declaration_{-v}))$.
    \item If $\beta$ is chosen and an object $\gamma$ is ejected, then
    $\SCF(\declaration_\voter';\Declaration_{-v}) = 
  \SCF(\declaration_\voter;\Declaration_{-v})-\{\gamma\}\cup \{\beta\}$. We have
  $\utility_\voter(\SCF(\declaration'_\voter,\Declaration_{-v})\cup\{\beta\}-\{\gamma\})\ge \utility_\voter(\SCF(\declaration_\voter,\Declaration_{-v})$
  because
    $\declaration^*_\voter(\beta)=1$, and $\declaration^*_\voter(\gamma)\le 1$. Thus,  $\utility_\voter(\SCF(\declaration_\voter',\Declaration_{-v}))\ge \utility_\voter(\SCF(\declaration_\voter,\Declaration_{-v}))$.
    \item Assume that the solution is not modified.  Then, $\SCF(\declaration_\voter';\Declaration_{-v}) = 
  \SCF(\declaration_\voter;\Declaration_{-v})$ and $\utility_\voter(\SCF(\declaration_\voter',\Declaration_{-v}))\ge \utility_\voter(\SCF(\declaration_\voter,\Declaration_{-v}))$.
    
    \end{enumerate}
    
In every case, $\utility_\voter(\SCF(\declaration_\voter';\Declaration_{-v}))\ge \utility_\voter(\SCF(\declaration_\voter;\Declaration_{-v}))\ge \utility_\voter(\SCF(\declaration_\voter^*;\Declaration_{-v}))$. The difference between $\declaration_\voter'$ and the sincere ballot $\declaration_\voter^*$, $|\declaration_\voter^* \backslash \declaration_\voter'|$, is $k$. By assumption $R_k$, $\utility_\voter(\SCF(\declaration_\voter^*,\Declaration_{-v}))=\utility_\voter(\SCF(\declaration_\voter';\Declaration_{-v}))\ge \utility_\voter(\SCF(\declaration_\voter;\Declaration_{-v}))$. Hence, $\utility_\voter(\SCF(\declaration_\voter^*,\Declaration_{-v}))= \utility_\voter(\SCF(\declaration_\voter;\Declaration_{-v}))$ and $R_{k+1}$ is true. 
    
With the induction, $\utility_\voter(\SCF(p_v;\Declaration_{-v}))=\utility_\voter(\SCF(\declaration_\voter^{*};\Declaration_{-v}))$ for every optimal ballot~$\declaration_v$.
\end{proof}

In the next section, we present results regarding the CCP property. In order to study the conditions on which the converse of Theorem~\ref{theo:4} does hold, we introduce the concept of \motnouveau{score functions}.

\section{Score Voting}
\label{sec:scorevoting}

We design an extension of the simple CM \SCMU\ model for the case where objects have unitary weight, which is always strategyproof. We introduce the notion of \motnouveau{score functions}, which allows us to introduce a correlation between objects: selecting an object will impact the selection of another object. Each time an object is taken, it can favour another object via a scoring process.


\begin{definition}[Score functions]
A \motnouveau{score function} on $m$ objects is defined by a couple ($M$, $\SCF$) with $m \times m$ real matrix $M$ and an associated social choice function $\SCF$.
Given  an integer $W$, and a ballot profile $\Declaration=(\declaration_1,\dots,\declaration_n)$, with $\declaration_1,\dots,\declaration_n$ seen as a column vector of $m$ elements, the social choice function $\SCF$ taking $\Declaration$ as input returns the winning set of $W$ objects which maximizes the inner product $M\cdot \mathbf{e}$ where the $i^{th}$ element  of  vector $\mathbf{e}$ represents the number of times where $\object_i$ is in  $\declaration_1$, $\dots$, $\declaration_n$. 
If such objects are not clearly defined, we use a tie break, a strict order given by the social choice function~$\SCF$.
\end{definition}

\begin{definition}[Score Voting]
A simple CM \SCM\ in which the social choice function is a score function is called a \emph{Score Voting}.
\end{definition}

Score Votings can be used to model utilitarian election systems where the voter must choose several options among a set of possible choices. The model is similar to Knapsack Voting~\cite{Goel19}, where every ballot that contains an object $\object$ gives one point to $\object$, and the winner is the object with the most points. The idea of a score function is to generalize the votes, allowing voters to give any amount of points to any given object

\noindent\textbf{Example:} The city council would like to propose four projects to the residents, such as ``renovating a library'' ($\object_1$), ``creating a bike path'' ($\object_2$), ``funding a soccer team'' ($\object_3$),   or ``funding a basketball team'' ($\object_4$). However, it can only fund two of the four, and she wants to avoid funding two sports-related projects. Thus, she can construct a score function with $M=  \begin{pmatrix}
3 & 0  & 0 & 0 \\
0 & 3  & 0 & 0\\
0 & 0  &3  &  -2\\
0 & 0  & -2 & 3\\
\end{pmatrix}$.

The ballot  distribution is the following: 
four voters want to fund only the sports projects, two voters want to fund only the non-sports projects, and two voters   want to fund the ``creating a bike path'', and ``funding a basketball team''. The following vector is summarized  by 
 $\begin{pmatrix}
2   \\
4  \\
4     \\
6 \\
\end{pmatrix}$.  Since $  \begin{pmatrix}
3 & 0  & 0 & 0 \\
0 & 3  & 0 & 0\\
0 & 0  &3  &  -2\\
0 & 0  & -2 & 3\\
\end{pmatrix}   \begin{pmatrix}
2   \\
4  \\
4     \\
6 \\
\end{pmatrix} =in \begin{pmatrix}
6  \\
12  \\
4 = (12-8)    \\
10 =  (18 - 8)\\
\end{pmatrix}$, the winner set is $\{$``creating a bike path'', ``funding a basketball team'' $\}$ because the scores of these objects are the highest.





We will study in the following sections the strategyproofness of score voting systems. We show that the strategyproofness of such systems can be characterized by their score function. In the remaining text, we will extend the use of the word strategyproof also to apply it to score functions that guarantee score voting systems to be strategyproof. In other words, we say that a score function is strategyproof if all score votings that use this function are strategyproof.

\subsection{Neutrality Property Over the Score Voting}


In this section, we show how to use score voting to design an election system where objects receive equal treatment. We use the standard matrix notations where  $M_{i,j}$ is the coordinate $(i,j)$ of matrix $M$, and $\mathbf{U}_i$ the coordinate $i$ of vector $\mathbf{U}$. The vector such that $(\delta_i)_i=1$ and $\forall j\neq i$, $(\delta_i)_j=0$, is denoted by  $\mathbf{\delta_i}$.

We start defining the \motnouveau{neutrality property} for election systems that treat objects equally:


\begin{definition}

A score function ($M$, $\SCF$) is \motnouveau{neutral} if  score matrix $M$ has the following form: $\forall k\in [m], i_1 \neq \dots \neq i_k \in [m], $
\begin{itemize}
\item$ M(\mathbf{\delta_{i_1}}+\dots+\mathbf{\delta_{i_k}})_{i_1}= \dots=
M(\mathbf{\delta_{i_1}}+\dots+\mathbf{\delta_{i_k}})_{i_k}$ ;
\item For all other coordinates $j\neq i_1,\dots,i_k$, $M(\mathbf{\delta_{i_1}}+\dots+\mathbf{\delta_{i_k}})_{j}$ is strictly less than  $ M(\mathbf{\delta_{i_1}}+\dots+\mathbf{\delta_{i_k}})_{i_1}$.
\end{itemize}
\end{definition}

For example, knapsack voting is a score voting respecting neutral property. Indeed, it is enough to notice that the identity matrix\footnote{The identity matrix is a $m\times m$ square matrix with ones on the main diagonal and zeros elsewhere.} of size $m$ $$ \begin{pmatrix}
  1 & 0 & \cdots & 0 \\
  0 & 1 & \cdots & 0 \\
  \vdots & \vdots & \ddots & \vdots \\
 0 & 0 & \cdots & 1 \end{pmatrix}.$$
with a tie break algorithm $\SCF$ corresponds to their score function.





We will now establish the set of score functions that are neutral. All of the following proofs are in the appendix.

\begin{proposition}\label{proprodile}
If ($M$, $\SCF$) is a neutral score function, then $M=D +\lambda_1 A^{(1)} +\dots +\lambda_m A^{(m)}_m$ with $\lambda_1,\dots,\lambda_m \in \mathbb{R}$, $D$ is a positive constant diagonal matrix and $\forall i,j,k \in [m], i\neq j$ $A^{(j)} $ is such that $A^{(j)}_{k,j}=1$ and $A^{(j)}_{k,i}=0$.
\end{proposition}

We can now compute the winning set of all scoring votes according to knapsack voting.


\begin{theorem} \label{th:neutral}
Every  neutral score voting has the same winning set as knapsack voting.
\end{theorem}

\subsection{Total score functions}

The previous results show that neutrality is a hypothesis that is too restrictive for election systems. If neutrality is satisfied, then there is only one possible winning set: the one of knapsack voting.

We use score functions to study a class of mechanisms that are strategyproof under a less restrictive hypothesis. For that, we introduce the notion of the totality of score functions:

\begin{definition}
A score function ($M$, $\SCF$) is called \emph{total} if there exists a ballot profile  $\Declaration=(\declaration_1,\dots,\declaration_n)$ such that   $\SCF(\declaration_1+\dots+\declaration_n)=s$, $\forall s \in \Si$. 
\end{definition}

Observe that Knapsack voting is not equivalent to total score functions (see Proposition~\ref{prop:total}).

\begin{proposition} \label{prop:total}
There exists a score voting such that the associated  score function ($M$, $\SCF$) is total and verifies  CCP property. It does not have the same winning set as  knapsack voting.
\end{proposition}
\begin{proof}
We can take as an example the matrix $
	\begin{pmatrix}
1 & 0 \\
0 & 3 
\end{pmatrix}$. ($M$, $\SCF$) is total. Indeed, if $W=1$, to obtain the winning set $\{o_i\}$, a ballot should be $\{i\}$. If $W=0$ or $W=2$, only one winning set is possible. Moreover, ($M$, $\SCF$) respects CCP since $M$ is diagonal. The winning set on the ballots $\{o_1\},\{o_1\},\{o_2\}$ is $\{o_1\}$ with knapsack voting and $\{o_2\}$ with ($M$, $\SCF$).
%
%
\end{proof}

From that, and from the fact that knapsack voting is total, we conclude that the set of total score functions strictly contains and is greater the set of neutral score functions. We now characterize the  strategyproofness of total score functions. To do so, we will define another set of constraints \emph{$\Delta$} on the matrix of score functions. We will prove for the restricted case of $m=3$ ($3$ objects) that strategyproofness is equal to CCP and to this new set \emph{$\Delta$} restricted to total functions.

\begin{definition}
We define $\Delta$ the following set of constraints on a score function ($M$,~$\SCF$):

 \begin{align*}
        \forall i,j,k,  \ \ \     & M_{i,i} - M_{i,j} \geq M_{k,i} - M_{k,j}       \\
         \forall i,j,j\neq i,  \ \ \                  & M_{i,j} < M_{j,j} \ if \  i>j   \\
         & M_{i,j} \leq M_{j,j}\  otherwise
    \end{align*}


We denote by \emph{$T_\Delta$} the set of all total score functions that respect $\Delta$.
\end{definition}

\begin{proposition}
The set $T_\Delta$ is not equal to the set of score functions that respect $\Delta$.
\end{proposition}
\begin{proof}
Let $M$ be 
$
	\begin{pmatrix}
4 & 0 & 0\\
0 & 4 & 0\\
3 & 3 & 100
\end{pmatrix}$. It respects $\Delta$ but is not total (with $W=2$, $\{\object_1,\object_2\}$ is not a feasible winning set).
\end{proof}

To prove that CCP and strategyproofness are equivalent whenever $m=3$, we will prove that:
\begin{itemize}
    \item any strategyproof total score function satisfies $\Delta$; and
    \item for any total score function ($M$, $\SCF$) that satisfies $\Delta$, $\SCF$ satisfies CCP.
\end{itemize}

We will first prove in Proposition~\ref{proposition:a1} that if ($M$, $\SCF$) is a strategyproof total score function, the equations $M_{i,j}<M_{j,j}$ are satisfied for all $i,j,j\neq i$. Then, we will use Lemma~\ref{lemma:a1} to prove in Proposition~\ref{proposition:a2} that if ($M$, $\SCF$) is a strategyproof total score function, the equations $ M_{j,k}-M_{j,i}\ge M_{k,k}-M_{k,i}$ are satisfied.

\begin{proposition} \label{proposition:a1}
If ($M$, $\SCF$) is a strategyproof total score function, then, for all $ i,j$, $j\neq i$: 
\begin{itemize}
    \item if $i$ is chosen over $j$ in the tie break of $\SCF$, then, $M_{i,j}<M_{j,j}$;
    \item otherwise, $M_{i,j}\le M_{j,j}$.
\end{itemize}
\end{proposition}

\begin{proof}
We build a score voting such that $W=1$ (the winning set is a singleton).

Let $i,j$ be two distinct integers in $[m]$  such that $M_{i,i}< M_{j,i}$ or $M_{i,i}= M_{j,i}$ and $j>i$ in the tie break.

Let $\mathbf{\delta_{i}}$ be a vector whose coordinates are $0$ except the $i$th coordinate which is $1$. By definition, the outcome (the winning set) on ballot $\mathbf{\delta_{i}}$ is not $\{o_i\}$: $\mathcal{A}lgo(M(\mathbf{\delta_i})) \neq \{o_i\}$.

Since ($M$, $\SCF$) is a total score function,  there exists a set of ballots $\Declaration=(\declaration_1,\dots,\declaration_n)$ such that the winning set $\mathcal{A}lgo(M(\Declaration))$ of the score voting on $\Declaration$ is $\{o_i\}$.


Several cases are possible. 
If $\declaration_1=\dots=\declaration_n=\mathbf{\delta_i}$, then $\mathcal{A}lgo(M(\declaration_1+\dots+\declaration_n))=\mathcal{A}lgo(M(\mathbf{\delta_i}))$. This leads to a contradiction with $\mathcal{A}lgo(M(\mathbf{\delta_i})) \neq \{o_i\}$.

Assume now that there exists a voter $k$ such that $\declaration_k \neq \mathbf{\delta_i}$. If voter $k$ wants object $\object_i$ and changes her ballot to $\mathbf{\delta_i}$, two cases are possible:
\begin{enumerate}
    \item The winning set is $\{o_i\}$. The new set of ballots is $\Declaration'=(\declaration_1,\dots,\declaration_{k-1},\mathbf{\delta_i},\dots,\declaration_n)$. In this case, another voter $k'$ such that $\declaration_{k'} \neq \mathbf{\delta_i}$ on the set $\Declaration'$ changes her ballot, and the process repeats. This process can be repeated at most $n$ times, until the set of ballot becomes $\Declaration''$ with $\declaration_1=\dots=\declaration_n=\mathbf{\delta_i}$, since
    the winning set on $\Declaration''$ is not $\{o_i\}$. Thus, after some steps, the other case must happen.
    \item The winning set is not $\{o_i\}$. Thus voter $k$ should vote $\declaration_k$. That shows that $M$ is not strategyproof.
\end{enumerate}

Hence, $M$ is either not total or not strategyproof. That shows that the contraposition of the proposition is true.
\end{proof}

The next lemma highlights the strength of total score functions. We show that given objects $x,y,z$ there exists a set of ballots such that $z$ gets the best score out of all of the winning sets, and $y$ gets a higher score than $x$.

We use this result to prove propositions \ref{proposition:a2} and \ref{proposition:p28} by contraposition.

\begin{lemma}\label{lemma:a1}

If ($M$, $\SCF$) is a total score function with a matrix of size at least $3 \times 3$ and $x,y,z$ are three objects, then  there exists a set of ballot $\Declaration=k_1\delta_1+\dots+k_x\delta_x+k_y\delta_y+k_z\delta_z+\dots+k_m\delta_m$ such that $M(\Declaration) =\begin{pmatrix}
t_1 \\
\dots \\
t_x \\
t_y \\
t_z \\
\dots
\end{pmatrix}$ with $t_z>t_y>t_x$, $t_z>t_i$ for every other object~$i$ and $k_1,k_2,k_3,\dots\neq 0$.
\end{lemma}

\begin{proposition} \label{proposition:a2} \label{proof:4}
For the total score functions, respecting CCP property implies respecting the set of constraints $\Delta$.
\end{proposition}


\begin{proposition} \label{proposition:3:CCP}
For the total score functions, with $3$ objects, respecting the set of equations $\Delta$ implies respecting property $CCP$.
\end{proposition}


The results show that for total score functions, the set of functions that respect CCP is included in the set of those that respect $\Delta$ and that they are equal whenever $m=3$.

Also, we can state the following important result regarding strategyproofness:

\begin{proposition} \label{proof:6}
Strategyproof total score functions respect  the set of constraints $\Delta$. 
\end{proposition}


Combining this proposition, Theorem \ref{theo:4} and Proposition  \ref{proposition:a2}, we can show that using a total score function with $m=3$, respecting CCP or being strategyproof are equivalent properties. Moreover, they show that voting mechanisms designed with a total score function respect the set of constraints $\Delta$ if and only if they are strategyproof.



Unfortunately, this result cannot be extended for $m>3$ objects.

\begin{proposition} \label{prop:936}
With strictly more than $3$ objects, there exist some score functions that respect the set of constraints  $\Delta$ and are not strategyproof.
\end{proposition}

\subsection{Total score functions with m>3}

In the previous section, we have defined the set of constraints $\Delta$ which helped us prove that CCP and strategyproofness are equivalent for total score functions whenever $m=3$. 

When $m>3$, the set of constraints $\Delta$ does not hold, even though the set $T_\Delta$ of total score functions that respect $\Delta$ is larger than the set of total score functions that are strategyproof. We study additional constraints for $\Delta$ that allow us to show the strategyproofness of some score voting systems with $m>3$. We denote this extended set of restrictions \emph{$\Delta^+$}.

\begin{definition}
Given a score function ($M$, $\SCF$), we define $\Delta^+$ is the following set of constraints that contains $\Delta$ and $\forall a,b,c,d \in \setObjects , M_{c,a}-M_{c,b}=M_{d,a}-M_{d,b}$.  
The set of total score functions that respect $\Delta^+$ is called \emph{$T_\Delta^+$.}
\end{definition}

\begin{proposition} \label{proposition:p28}
If $m>3$, every strategyproof total score function satisfies $\Delta^+$.
\end{proposition}
 
\begin{proposition} \label{proposition:p29}
If $m>3$, every total score function that satisfies $\Delta^+$ satisfies CCP.
\end{proposition}
 
From the result when $m=3$, and  $m>3$:

\begin{theorem}
A total score function is strategyproof if and only if it satisfies CCP.
\end{theorem}

\subsection{Finding the closest strategyproof total score function}

The set of total score functions $T_\Delta^+$ can be used to devise an algorithm capable of computing the closest (with respect to Frobenius norm\footnote{$\|A\|_\text{F} = \sqrt{\sum_{i,j=1}^n |a_{ij}|^2} = \sqrt{\operatorname{Tr}\left(A^* A\right)}$}~\cite{golub2013matrix} on matrix) strategyproof total  score function to another given score function, if it exists.
Since $T_\Delta^+$ is a convex polyhedron, quadratic optimization helps us to project the score function on $T_\Delta^+$.  

An issue is that some of the inequations of $\Delta$ are strict, but the projection is done on the closed set $\overline{T_\Delta^+}$. This leads to two possibilities:
\begin{enumerate}
    \item the projection gives a strategyproof total score function;
    \item the notion of the closest strategyproof total score function does not exist, and the projection gives a non-strategyproof score function.
\end{enumerate}

In the first case, the projection is the result of the algorithm. We now show a way to handle the second case.

We project a matrix $M$ on $\overline{T_\Delta^+}$, which gives a matrix $M'$. If $M'$ is not strategyproof, we cover the sphere centered in $M'$, of radius $\delta$ (a small constant) with a set of points, so that every point on the sphere is at most at a distance $\epsilon$ to a point of the cover. To cover such a sphere, in dimension $n^2$, we consider the points $(k_1\epsilon,k_2\epsilon,\dots,k_{n^2-1}\epsilon,\sqrt{1-k_1^2\epsilon^2-\dots-k_{n^2-1}^2\epsilon^2})$ $\forall k_1,\dots,k_{n^2-1} \in [-\frac{1}{\epsilon},\frac{1}{\epsilon}]$ whenever the square root is defined.
With this cover, when $\epsilon$ is close to $0$, the distance between a point of the cover, and a point of the sphere is at most $n^2*\sqrt{2\epsilon}$.

Using the inequations of $\Delta^+$, we can check on every point of the cover whether they are strategyproof or not and output a ``close'' strategyproof matrix. This procedure gives an algorithm that outputs the closest strategyproof total score function if it exists. It outputs a strategyproof total score function at distance $\delta$ from the optimal if the closest matrix does not exist.

Given a score function ($M$, $\SCF$), the algorithm can be described as follows:

\begin{enumerate}
    \item Solve the quadratic optimization problem to project $M$ on $\Delta^+$;
    \item Verify if the projection ($M'$, $\SCF$) respects the constraints of $\Delta^+$ and if so, return ($M'$, $\SCF$);
    \item Otherwise, find and return a score function that respects the constraints of $\Delta^+$ on the sphere centered in $M'$, of radius $\epsilon$.
\end{enumerate}

\begin{theorem}
Given a score function ($M$, $\SCF$), there exists an algorithm that returns the closest strategyproof total score function.
\end{theorem}

\section{Conclusion and perspectives} \label{sec:conclusions}

In this paper, we studied the design of strategyproof election systems. We focused on voting systems like the ones used in Participatory Budgeting, where one must choose a fixed number of options among several possibilities (e.g., projects to be funded restricted to a given budget).

Using a new model for describing choice mechanisms, we detailed the conditions that make non-dictatorial mechanisms non-strategyproof. We also presented the Constrained Change Property (CCP), a powerful tool to characterize the strategyproofness of voting mechanisms. Using CCP, we showed how to design choice mechanisms that are strategyproof for the unitary case.

We proposed the notion of score functions and a new class of utilitarian voting mechanisms called Score Voting. Score functions turned out to be a flexible tool to characterize the strategyproofness of voting mechanisms. 
We proved that the mechanisms designed with score voting with a neutral score function are equivalent to knapsack voting on the same instance. Also, if the score function satisfies CCP, then, the score voting is strategyproof.
These results were combined to devise an algorithm that can find the closest total score function that makes any given score voting to be strategyproof. 

In future works, we intend to further improve the notion of score voting (and its score functions) to better characterize the relation of the CCP property and strategyproofness. We are also interested in the problem of quantifying the strategyproofness of a vote and in creating approximation algorithms for the problem.

\bibliography{bibliography}


\appendix

\section{Appendix: omitted proofs and examples}
\label{appendix}



\subsection{Proofs of Section~\ref{sec:impossibility}}

\begin{proposition}\label{proposition:p1}
Given a simple CM  \SCMU, a polynomial time algorithm can compute the winning set of objects in the unitary case. 
\end{proposition}
\begin{proof}

This problem can be seen as the problem of the knapsack: the objects are the same size, and the value of an object $\object$ is equal to $\sum_{\voter \in \setObjects} \declaration_{\voter}(\object)$. In this context, the problem of the knapsack can be solved with a greedy algorithm.
\end{proof}

The NP-completeness proofs are based on the two NP-complete problems:

 \begin{definition*}
Problem $Knapsack(\setObjects,w,u,W, T)$\newline
\noindent{\bf Input:} A set $\setObjects$ of objects with a value function $u: \setObjects  \to \mathbb{N}^+$ and with a weight function $w: \setObjects  \to \mathbb{N}^+$, a capacity $W$ of the knapsack, an integer  $T$.\newline
\noindent{\bf Question:} Does there exist a set S such that  $ \sum_{\object \in S } u(S) \ge T $ and $\displaystyle\sum_{\object \in S }w(\object) \le W$?

\end{definition*}
 
 \begin{definition*}
Problem Vertex Cover $VC(G,k)$\newline
\noindent{\bf Input:} A graph $G=(V,E)$ where $V$ is the vertex set and $E$ is an edge set, a positive integer $K\le |V|$.\newline
\noindent{\bf Question:} Is there a vertex cover of size $k$ or less for $G$ (a subset $S$ of $V$ of size at most $K$ such that for every $\{u,v\}\in E$, $u\in S$ or $v\in S$)
\end{definition*}

As shown by 
Karp~\cite{Karp72}, problems Knapsack and Vertex Cover are NP-complete.

\begin{proposition}\label{proposition:p7}
The following problem called $BvMm_1$:

\noindent{\bf Input:} A set of voters $\setVoters$, a set of objects $\setObjects$, ballots: $\setObjects \to \{0,1\}$ with $T$ images of value $1$.

\noindent{\bf Question:} Does there exist a  set $S$ of $\setObjects$  such that  $  min\{\utility_\voter(S) \mid \voter \in \setVoters\} \ge 1 $ and $ |S|\le W$?

is NP-complete.
\end{proposition}

\begin{proof}
It is easy to see that $BvMm_1$  is in NP because a nondeterministic algorithm only needs to guess a subset $S$  of objects and check in polynomial time whether $  min\{\utility_\voter(S) \mid \voter \in \setVoters\} \ge 1 $ and $ |S|\le W$.

We transform an instance of Vertex Cover into an instance of Problem $BvMm_1$.  Let $G=(Vertices,Edges)$ and $k$
be an instance of Vertex Cover. We must construct an instance of Problem $BvMm_1$.

The set $\setObjects$ of objects is the set of vertices $Vertices$ and $k-2$ distinct new objects for every edge $e$ in $Edges$. We call $\mathbf{e}_s$ this set of the new objects for edge $e$. Every  object  $\object$ has its weight $w(\object)$ equal to $1$. Finally, we set $k$ to $W$. For every edge $e=(a,b)\in Edges$, we create a voter $\voter_e$ which votes for $a,b$, and all objects in  $\mathbf{e}_s$: $\declaration_{\voter_e}=\{a,b\}\cup \mathbf{e}_s$. 

This construction can be built in polynomial time.

Assume that there is a solution $S$ for the instance built of $BvMm_1$.

We assume that $S$ only contains objects corresponding to a vertex. If it is not the case, we can modify $S$ as follows.
If $S$ contains an object $\object$ of $\mathbf{e}_s$ for some $e=(a,b)$,  we build a new 
set $S'$ such that $S'= S - \{object\} \cup \{a\} $.  Observe that $\utility_{\voter_{e'}}(S')= \utility_{\voter_{e'}}(S')$ for all $e' \in Edges$ and $  min\{\utility_\voter(S') \mid \voter \in \setVoters\} \ge 1 $. So  $S'$ is also a solution for an instance built of $BvMm_1$.


Since $ |S|\le W$, at most $W(=k)$ objects are chosen. Also, since for every voter $\voter_{e}$, $\utility_{\voter_{e}}(S) \geq 1$,
 at least one extremity of edge $e$ is in  $S$. Thus, $S$ is also a vertex cover of $G$.
 
Conversely, assume that graph $G$ admits a vertex cover $S$ of size at most $k$. We prove that $S$ is a solution for the instance of $BvMm_1$. Since $S$ is a vertex cover, at most, one extremity of each edge is in $S$, then 
 $\utility_{\voter_{e}}(S) \geq 1$ for all $v_e \in \setVoters$. Thus $S$ is a solution of our instance of $BvMm_1$. This proposition holds. \end{proof}




\begin{proposition}\label{proposition:p9}
The following problem MaxMinRankingBallot \\
\noindent{\bf Input:} A set of voters $\setVoters$, a set of objects $\setObjects$, ballots which are bijections: $\setObjects \to [m]$ , a positive integer $T$.
\newline
\noindent{\bf Question:}  Does there exist a subset $S$ of $\setObjects$ such that  $  \min\{\utility_\voter(S): \voter \in \setVoters\} \ge T $ and $ |S|\le W$\newline
 is NP-complete.\end{proposition} 

\begin{proof} 

It is easy to see that MaxMinRankingBallot is in NP  because a nondeterministic algorithm only needs to guess a subset $S$  of objects and check in polynomial time whether  $ \max \min\{\utility_\voter(S): \voter \in \setVoters\} \ge T $ and $ |S|\le W$.

Recall that Problem $BvMm_1$ is NP-complete due to Proposition~\ref{proposition:p7}.

Let us take an instance $S$ of $BvMm_1$ and create an instance that for every ballot, has $W+1$ times this ballot (in the end, we get $W+1$ times the ballots of the first instance). We get an instance $S^*$ which has the same solution as the one of $S$ (By multiplying (resp dividing) the value of a solution of $S$ (resp $S^*$), we get a solution of the other instance. This transforms an optimal solution of one into an optimal solution of the other).\newline
Now, let us create an instance $S^{**}$ of MaxMinRankingBallot. The objects and maximum weight $W$ are the same as those of $S^*$. For every voter $\voter_i$ of $S^*$, we create two voters $\voter_{i,1}$ and $\voter_{i,2}$. Lets take an order on the $W$ objects $\voter_i$ votes for: $(o_1,\dots,o_W)$. $\voter_{i,1}(o_i)=i$, $\voter_{i,2}(o_i)=W+1-i$, they give value $0$ to the other objects. $\voter_i(o)=\voter_{i,1}(o)+\voter_{i,2}(o)$, hence, the optimal solutions of $S$ and $S^{**}$ are the same. This transformation is clearly polynomial, so MaxMinRankingBallot is NP-complete.
\end{proof}

\begin{proposition}\label{proposition:p19}
The following problem  \\
\noindent{\bf Input:} A simple CM \SCMU, the ballot is a valuation of the objects of $\setObjects$ (a function $\declaration_\voter:\setObjects\to \mathbb{R}$). 

\noindent{\bf Question:}  Does there exist a subset $S$ of $\setObjects$  such that $\displaystyle\sum_{\voter \in \setVoters }\utility_\voter(S)\ge T$ and such that and $ |S|\le W$

 is NP-complete.
\end{proposition} 
\begin{proof}
This problem is clearly in NP. Any instance of a knapsack problem is an instance of this Problem considering a single voter. The desired result holds.
\end{proof}

\begin{proposition}\label{proposition:p13}
Given a simple CM  \SCMU\  where the ballot of each voter is defined as a subset of  $\setObjects$ (approval voting), a polynomial time algorithm can compute a set $S$ of  $\setObjects$ such that $max\displaystyle\sum_{\voter \in \setVoters }\utility_\voter(S)\ge T$ and such that $\displaystyle\sum_{\object \in S }\weight(\object) \le W$.

\end{proposition}
\begin{proof}

Let us first give the idea of the algorithm. 
The Knapsack algorithm (dynamic programming) inspires our algorithm. Using the terminology, we compute the minimal capacity of a knapsack that can contain a subset $S$ of objects that maximizes the function $\utility$.

The first step is to sum up all the values given to the object to calculate the value $\mu=\max_{\object \in \setObjects} \sum_{\voter \in \setVoters} \utility_\voter(\object) $. Then, the algorithm proceeds in a dynamic way. We create a dynamic vector $\mathbf{dynamic\_weight}$ to have more entries than the values of all possible solutions. This vector has $\mu*m$ elements.

We initialize this vector with $0$ on the first entry $0$ and $\infty$ everywhere less. The value of an entry $\ell$ represents the minimum total weight required to get, with a set $S$, a score $\sum_{\voter \in \setVoters }\utility_\voter(S)$ of  $\ell$. 

We then fill this vector in a dynamic way, with a loop on the objects (at each step of this loop, we update the vector).

\begin{algorithm}
\caption{Algorithm for the multiplayer knapsack voting with raking approval ballots.}\label{alg:cap1}
 \KwData{A set of objects $\setObjects$, a set of voters $\setVoters$ with ballots: $\setObjects \times \setVoters\to \{0,1\}$, a maximum weight $W$, the ballot being a approval voting.}
 \KwResult{ A set S that maximizes  $max\displaystyle\sum_{\voter \in \setVoters }\utility_\voter(S)$ and such that $\displaystyle\sum_{\object \in S }\weight_\object \le W$.}

 \For{every object $\object_i$}{
 $values[i] \gets \sum_{\voter \in \setVoters} \utility_\voter(\object_i)$ \; 
 
 }
 $\mu \gets max(\{values[i] : \object_i \in \setObjects\})$ \;
  $weight_{max} \gets max(\{\weight(\object_i) : \object_i \in \setObjects\})$ \;

$\mathbf{dynamic\_weight} \gets \{0,weight_{max}*m+1,\dots,weight_{max}*m+1\}$ a list of length $\mu * m$ \;
$dynamic\_outcome \gets \{\{\},...,\{\}\}$ a list of length $\mu * m$ \;

 \For{$i$ varying from $1$ to m}
 { \For{$k$ varying from $\mu * m$ to $1$ decreasingly}{
 $j \gets k-values[i]$\;
 \If{$j<0$}{break \;}
 \Else{
 \If{$\mathbf{dynamic\_weight}[j]+\weight(\object_i)<\mathbf{dynamic\_weight}[k]$ and $\mathbf{dynamic\_weight}[j]+\weight(\object_i)<W$}{
 
 $\mathbf{dynamic\_weight}[k] \gets \mathbf{dynamic\_weight}[j]+\weight(\object_i)$\;
 $\mathbf{dynamic\_outcome}[k] \gets \mathbf{dynamic\_outcome}[j]+\{\object_i\}$\;
 }
 
 }
 
 }
 }
 
 \Return The highest entry $k$ of $\mathbf{dynamic\_outcome}$ such that $\mathbf{dynamic\_weight}[k]$ is different to $weight_{max}*m+1$\;

\end{algorithm}

Let us prove Algorithm~\ref{alg:cap1}.

\noindent{Correctness:} 
In the double `for' loop, an invariant is: At the $i$th step, for every $\ell$, $\mathbf{dynamic\_weight}[\ell]$ contains the least weight necessary to get value $\ell$ only using the $i$-th first objects (or $weight_{max}*m+1$).\newline
Indeed, when $\mathbf{dynamic\_weight}$ is created, it is true. At every step of the loop on $i$, every possible value is updated according to the object $i$: if using this object always gives a higher value, the algorithm looks at the next value, and using the recurrence, the property is still verified. If using this object requires having a certain value $j$ from another item to reach the target value, then the weight of $\mathbf{dynamic\_weight}[value]$ is updated based on which one is the lowest (using object $i$, or not).

$weight_{max}*m+1$ is not a reachable weight, and if the value is different than $weight_{max}*m+1$ in $\mathbf{dynamic\_weight}$ at the end of the loops, then it is an attainable value (based on the invariant). Hence, the highest value, with a weight different than $weight_{max}*m+1$ in $\mathbf{dynamic\_weight}$ is the highest value attainable. The sets inside the $\mathbf{dynamic\_outcome}$ are modified when necessary so that the scores of these sets are equal to those of the entry. Hence, the set in the highest entry attainable is optimal.

\noindent{\bf Complexity:} The complexity of the first loop is $O(|\setObjects|)$. The complexity of the combination of two for loops is $O(|\setObjects|^2*\mu)$. The complexity of searching the minimum in the table is $O(|\setObjects|*\mu)$. Thus, the complexity of the algorithm is $O(|\setObjects|^2*\mu)=O(|\setObjects|^2*|\setVoters|)$
\end{proof}

\subsubsection{About strategyproofness }


\begin{proposition}\label{proposition:p14}\label{proposition:SP:AV:utilitarian:general}
Considering the approval voting \SCMF\ when all objects have the same weight, the social choice function   that returns the optimal solution is not strategyproof. 
\end{proposition} 

\begin{proof} 
Let us consider an instance with $W=60$, $6$ voters and $9$ objects $a,b,c,d,e,f,g$. $\weight(a)=\weight(b)=\weight(f)=5$, $\weight(c)=10$, $\weight(d)=20$, $\weight(e)=15$, $\weight(g)=40$. We define the preference ranking as follows using approval voting:
\begin{itemize}
\item    $a,e,g$  for voter $\voter_1$;  
\item  $d,g$   for voter $\voter_2$;  
\item   $a,e,g$  for voter $\voter_3$;  
\item  $d,g$ for voter $\voter_4$;  
\item   $b,c,f,g$ for voter $\voter_5$;  
\item   $a,e,g$ for voter $\voter_6$.
\end{itemize}

Let us consider the instance in which all voters vote sincerely:

\begin{tabular}{ |p{2cm}|p{1cm}|p{1cm}|p{1cm}|p{1cm}|p{1cm}|p{1cm}|p{1cm}| }
 \hline
      & $a$  & $b$ & $c$ & $d$ & $e$ & $f$ & $g$  \\\hline
    $\voter_1$  & $1$  & $0$ & $0$ & $0$ & $1$ & $0$ & $1$\\\cline{1-8}
    $\voter_2$  & $0$  & $0$ & $0$ & $1$ & $0$ & $0$ & $1$\\\cline{1-8}
    $\voter_3$  & $1$  & $0$ & $0$ & $0$ & $1$ & $0$ & $1$\\\cline{1-8}
    $\voter_4$   & $0$  & $0$ & $0$ & $1$ & $0$ & $0$ & $1$\\\cline{1-8}
    $\voter_5$  & $0$  & $1$ & $1$ & $0$ & $0$ & $1$ & $1$ \\\cline{1-8}
    $\voter_6$  & $1$  & $0$ & $0$ & $0$ & $1$ & $0$ & $1$\\\cline{1-8}
    total  & $3$  & $1$  & $1$ & $2$ & $3$ & $1$ & $6$  \\\cline{1-8}
 \hline
\end{tabular} 

The set $\{a,e,g\}$ is the best one and is selected. 

So the utility of each voter are: $\utility_1(\{a,c,h\})=3$, $\utility_2(\{a,c,h\})=1$, $\utility_3(\{a,c,h\})=3$, $\utility_4(\{a,c,h\})=1$, $\utility_5(\{a,c,h\})=1$, $\utility_6(\{a,c,h\})=3$.

Assume voter $5$ changes her vote to $\textbf{a},b,c,\textbf{d,e},f$. Now the scores change and are:

\begin{tabular}{ |p{2cm}|p{1cm}|p{1cm}|p{1cm}|p{1cm}|p{1cm}|p{1cm}|p{1cm}| }
 \hline
      & $a$  & $b$ & $c$ & $d$ & $e$ & $f$ & $g$  \\\hline
    $\voter_1$  & $1$  & $0$ & $0$ & $0$ & $1$ & $0$ & $1$\\\cline{1-8}
    $\voter_2$  & $0$  & $0$ & $0$ & $1$ & $0$ & $0$ & $1$\\\cline{1-8}
    $\voter_3$  & $1$  & $0$ & $0$ & $0$ & $1$ & $0$ & $1$\\\cline{1-8}
    $\voter_4$   & $0$  & $0$ & $0$ & $1$ & $0$ & $0$ & $1$\\\cline{1-8}
    $\voter_5$  & $\textbf{1}$  & $1$ & $1$ & $\textbf{1}$ & $\textbf{1}$ & $1$ & $\textbf{0}$ \\\cline{1-8}
    $\voter_6$  & $1$  & $0$ & $0$ & $0$ & $1$ & $0$ & $1$\\\cline{1-8}
    total  & $\textbf{4}$  & $1$  & $1$ & $\textbf{3}$ & $\textbf{4}$ & $1$ & $\textbf{5}$  \\\cline{1-8}
 \hline
\end{tabular}

The set $\{a,b,c,d,e,f\}$ is the best one and is selected. 

So the utility of each voter are: $\utility_1(\{a,c,h\})=2$, $\utility_2(\{a,c,h\})=1$, $\utility_3(\{a,c,h\})=2$, $\utility_4(\{a,c,h\})=1$, $\utility_5(\{a,c,h\})=3$, $\utility_6(\{a,c,h\})=2$.

The utility of $\voter_5$ increases. Hence, she has incentives to give a non-sincere vote. Thus, this mechanism is not strategyproof.\end{proof}

\begin{proposition}\label{proposition:p4}\label{proposition:SP:RV:utilitarian:general}
There exists a ranking voting \SCMU\  without restriction on weights such that the social choice function that returns the optimal solution is not strategyproof. 
\end{proposition}

\begin{proof}
Gibbard-Satterthwaite's theorem applies when $W=1$.
\end{proof}

\begin{proposition}\label{proposition:p6}\label{proposition:SP:VF:utilitarian:unitary}
There exists a value function voting \SCMU\ when all objects have the same weight  such that the social choice function   that returns the optimal solution is not strategyproof. 
\end{proposition} 
\begin{proof}
The mechanism may simulate a ranking voting, which is already not strategyproof.
\end{proof}

\begin{proposition}\label{proposition:p8}\label{proposition:SP:AV:fair:unitary}
There exists an approval voting \SCMU\ when all objects have the same weight  such that the social choice function   that returns the optimal solution is not strategyproof. 
\end{proposition} 

\begin{proof}  We build the strategyproof approval voting.
There are  $4$ objects $a,b,c,d$ and $2$ voters (even if the objects chosen by everyone are chosen in the end). The weight constraint $W$ is equal to $2$.  
First : if we have an order $a>b>c>d$ on the objects.
We define the approval voting with sincere ballots as follows:

\begin{itemize}
\item$v_1$:  $a,c$
\item$v_2$:  $a,b$
\end{itemize}
With this ranking, according to the order, the best solution is $a,b$, thus, this solution is chosen.\newline

Assume voter $2$ changes her vote to $c,\textbf{d}$. \newline
With this ranking, according to the order, the best solution is $a,c$ (because with the optimization function, $a,b$ has score $0$ and so cannot be chosen), thus, $a,c$ is chosen, increasing the utility of player $1$ from $1$ to $2$.\end{proof}

\begin{proposition}\label{proposition:p24}\label{proposition:SP:RV:fair:general}
There exists a ranking voting \SCMF\  without restriction on weights such that the social choice function that returns the optimal solution is not strategyproof. 
\end{proposition} 
\begin{proof} 
If all the weights are the same and $W=1$, then Gibbard-Satterthwaite's theorem applies.
\end{proof}

\begin{proposition}\label{proposition:p10}\label{proposition:SP:RV:fair:unitary}
 There exists a ranking voting \SCMF\  with equal weights (unitary) such that the social choice function that returns the optimal solution is not strategyproof.

\end{proposition} 
\begin{proof}
 Gibbard-Satterthwaite's theorem applies when $W=1$.
\end{proof}

\begin{proposition}
\label{proposition:p12}\label{proposition:SP:VF:fair:unitary} 
There exists a value function voting \SCMF\  with equal weights (unitary) such that the social choice function that returns the optimal solution is not strategyproof. 
 \end{proposition}
 \begin{proof}
 It is already not strategyproof if the value function simulates a ranking. 
\end{proof}

\begin{proposition*}{\bf  \ref{proposition:p15}} \label{proposition:SP:AV:utilitarian:unitary:Gibbard} 
There exists an approval voting \SCMF\  with $\utility_\voter(e)=\sum_{\object\in e\cap d^*_\voter}\weight(\object)$ such that the social choice function that returns the optimal solution is not strategyproof. 
\end{proposition*}

\begin{proof}
We are looking at a strongly non-dictatorial first-class CM, the utility can take more than three values, and if we create an object of weight. We can lower the utility of a solution by removing an object from the sincere vote, it does not change the utility of the other solutions that are not included in the first one if the object is selected the right way. Hence, the proof of \ref{theo:3} adapts and this CM is not strategyproof.
\end{proof}

\begin{proposition}\label{proposition:p17}
The following problem

\noindent{\bf Input:} A set of voters $\setVoters$, a set of objects $\setObjects$, ballots which are bijections: $\setObjects \to [m]$ , a positive integer $T$.

\noindent{\bf Output:} A set $S$ such that $max\displaystyle\sum_{\voter \in \setVoters }\utility_\voter(S)\ge T$ and such that $\displaystyle\sum_{\object \in S }\weight(\object) \le W$.

can be solved in  polynomial time.
\end{proposition}
\begin{proof}
The algorithm in the proof of  Proposition \ref{proposition:p13} still works in this case.
\end{proof}



\begin{proposition}\label{proposition:p11}
The following problem $MMmK_1$ 

\noindent{\bf Input:} A set of voters $\setVoters$, a set of objects $\setObjects$, ballots: $\setObjects\to\mathbb{R}$, a positive integer $T$.

\noindent{\bf Question:} Does there exists  a set S such that  $max min\{\utility_\voter(S): \voter \in \setVoters\} \ge T $ and $ |S|\le W$?

 is NP-complete.
\end{proposition}

\begin{proof}

First, it is easy to see that \emph{MMmK}$_1$ is in NP.  
A non-deterministic algorithm can exhibit a solution and check whether a subset of $\setObjects$ is a solution.

Now, we transform  \emph{PARTITION} to $MMmK_1$. We recall the definition of this NP-complete problem~\cite{garey1979computers}:

\noindent{\bf Input:} A finite set $A$ of integers.
 
\noindent{\bf Question:} Is there a subset  $A' \subseteq A$ such that $\sum_{a \in A'} a = \sum_{a \in A\backslash A'} a $ and $|A'| = |A\backslash A'|$?

Let us take an instance $A$ of the  \emph{PARTITION} problem. Then, we create $n+2$ voters and $|A|$ objects.  Thus $\setObjects = \{ \object_j : j \in A \} $. Every object $object$ has its weight equal to $1$ ($\weight(\object)=1$),  $W=|A|/2$, and $T= 1/2 \sum_{a \in A} a$.  We build the utility.

\begin{itemize}
\item For voter $\voter_{1}$, $\utility_{\voter_{1}}(\object_{j})=a_{j}$, for every $j \le |A|$.
\item For voter $\voter_{2}$, $\utility_{\voter_{2}}(\object_{2j})=a_{2j+1}$ for every $j$ such that $2j\le |A|$,  and $\utility_{\voter_{2}}(\object_{2j+1})=a_{2j}$, for every $j$ such that  $2j+1\le |A|$. 
\item  For all other voters $\voter_{i}$, ($i\ge 3$), $\utility_i(\object_{j})=0$ for every object $object_{j}$ with $j \le |A|$,  except for $\utility_i(\object_{2\cdot  (i-2)})=\utility_i(\object_{2\cdot  (i-2)+1})=(\sum_{a\in A} a)$.
\end{itemize}


This transformation can be performed in polynomial time.

Assume  problem~\emph{PARTITION}  has a solution $A'$.  Let 
$\setObjects'$ be the set $\setObjects'$ of the $\frac{|A|}2$ objects such that 
$\object_k \in \setObjects'$ if $j\in A'$. The utility of  voters $\voter_{1}$ and $\voter_2$ gives $\sum_{\object\in \setObjects'}\utility_\voter(\object)=\frac{\sum_{a\in A}a}2$, and  the utility of  all the other voters  is $\sum_{a\in A} a$.

Conversely, assume there is a set of objects  
$\setObjects'$  such that the utility of all the voters 
for $\setObjects'$ is greater than $T=\frac{\sum_{a\in A}a}2$.

Every voter gives the selected set of objects the value $\sum_{\object\in \setObjects'}\utility_i(\object)=\frac{\sum_{a\in A}a}2$. Then, for every one of the $n$ last voters, the object $\object_{2k-1}$ or the object $\object_{2k}$ has to be selected. Hence, the selected set of objects respects the
constraint: ``the elements of $A$ are ordered as $a_1,\dots,a_{2n}$ and we require that $A'$ contains exactly one of $a_{2k-1}$ or $a_{2k}$ for $1\le k\le n$''. Plus, because of the $n/2$ voters $\object_{2k-1}$ or $\object_{2k}$ selected, for the two first voters, $\sum_{a\in A'}a=\sum_{a\in A\backslash A'}a$. Thus, there is a solution to the problem~\emph{PARTITION}.

Thus the desired result follows.
\end{proof}

\subsection{Proofs of Section~\ref{sec:scorevoting}}

\begin{proposition*}{\bf \ref{proprodile} }
If ($M$, $\SCF$) is a neutral score function, then $M=D +\lambda_1 A^{(1)} +\dots +\lambda_m A^{(m)}_m$ with $\lambda_1,\dots,\lambda_m \in \mathbb{R}$, $D$ is a positive constant diagonal matrix and $\forall i,j,k \in [m], i\neq j$ $A^{(j)} $ is such that $A^{(j)}_{k,j}=1$ and $A^{(j)}_{k,i}=0$.
\end{proposition*}

\begin{proof}\label{proof:2}
We will prove this result by recurrence on the size of the matrix, i.e., the number of objects of the vote. In this proof, we will also prove this result in the case of only $2$ objects (it helps with the recurrence).

If $M$ is the $2 \times 2$ left corner submatrix, then $M_{1,1}\ge M_{2,1}$ because $\forall i ~ \mathcal{A}lgo(\delta_i)=i$. $M(\delta_{1}+\delta_{2})_1=M(\delta_{1}+\delta_{2})_2$ so $M_{1,1}+M_{1,2}=M_{2,1}+M_{2,2}$ so, by taking $c=M_{1,1}-M_{2,1}$, then, $M_{2,2}-M_{1,2}=c$, hence, with $\lambda_1=M_{2,1}$ and $\lambda_2=M_{1,2}$ we proved the result for $n=2$.

If $M$ is the $n+1*n+1$ left corner submatrix, then let us suppose that the proposition is true for the $n*n$ left corner submatrix.

$M(\delta_{1}+\dots+\delta_{i_{n+1}})_1=\dots=M(\delta_{1}+\dots+\delta_{i_{n+1}})_{n+1}$. Hence, for $k$ such that $2 \le k<n+1$, $M_{1,1}+\dots+M_{1,n+1}=M_{k,1}+\dots+M_{k,n+1}$. But due to the result at rank $n$, $M_{1,1}+\dots+M_{1,n}=M_{k,1}+\dots+M_{k,n}$. Hence, $M_{1,n+1}=M_{k,n+1}=\lambda_{n+1}$.

Now take $k$ such that $1 \le k<n+1$; $M(\delta_{k}+\delta_{i_{n+1}})_k=M(\delta_{k}+\delta_{i_{n+1}})_{n+1}$ 
so $M_{k,k}+M_{k,n+1}=M_{n+1,k}+M_{n+1,n+1}$, $c+\lambda_k+\lambda_{n+1}=M_{n+1,k}+M_{n+1,n+1}$. Let us take $c'=M_{n+1,n+1}-\lambda_{n+1}$. Now, $M_{n+1,k}=\lambda_k+c-c'$.

Due to $M(\delta_{1}+\dots+\delta_{i_{n+1}})_1=\dots=M(\delta_{1}+\dots+\delta_{i_{n+1}})_{n+1}$,
$M_{1,n+1}+\dots+M_{n+1,n+1}=M_{1,1}+\dots+M_{n+1,1}=c+\lambda_1+\dots+\lambda_{n+1}$. 
$\lambda_1+\dots+\lambda_{n+1}+n*c-(n-1)c'=c+\lambda_1+\dots+\lambda_{n+1}$. Hence, $c=c'$, so $M_{n+1,n+1}=c+\lambda_{n+1}$ and $M_{n+1,k}=\lambda_k+c-c'=\lambda_k$, which proves the proposition at rank $n+1$
\end{proof}

\begin{theorem*}{\bf \ref{th:neutral}}
If ($M$, $\SCF$) is a neutral score function, then the winning set of the social score function is the same as knapsack voting ($I_m$).
\end{theorem*}

\begin{proof}\label{proof:3}
Let us choose two objects $a$ and $b$ and $@a$ (resp $@b$) the number of times that $a$ (resp $b$) appears in the ballots. If $@a=@b+k$ with k a constant, due to $M(@a,@b)_a=M(@a,@b)_b$ and that voting for another object $c$ gives the same score to $a$ and $b$ (Proposition \ref{proprodile}), the score of $a$ is higher than the one of $b$. Hence, the winners are the objects with the highest amount of vote. This is exactly the mechanism of knapsack voting.
\end{proof}

\begin{lemma*}{\bf~\ref{lemma:a1}}
If ($M$, $\SCF$) is a total score function with a matrix of size at least $3 \times 3$ and $x,y,z$ are three objects, then  there exists a set of ballot $\Declaration=k_1\delta_1+\dots+k_x\delta_x+k_y\delta_y+k_z\delta_z+\dots+k_m\delta_m$ such that $M(\Declaration) =\begin{pmatrix}
t_1 \\
\dots \\
t_x \\
t_y \\
t_z \\
\dots
\end{pmatrix}$ with $t_z>t_y>t_x$, $t_z>t_i$ for every other object~$i$ and $k_1,k_2,k_3,\dots\neq 0$.
\end{lemma*}

\begin{proof}
$M$ is total, hence there exists $k_{x,1},k_{y,1},k_{z,1},\dots,k_{m,1}$ such that \newline \textcolor{black}{\fbox{$M(k_{x,1}\delta_x+k_{y,1}\delta_y+k_{z,1}\delta_z+\dots+k_{m,1}\delta_i)=\{z\}$}}.\newline\newline

$M$ is total, hence there exists
  $k_{x,2},k_{y,2},\dots,k_{m,1}$ such that\newline
  \textcolor{black}{\fbox{$M(k_{x,2}\delta_x+k_{y,2}\delta_y+k_{z,2}\delta_z+\dots+k_{m,2})=\{y,z\}$}}
  \newline
  
   (\textcolor{black}{\fbox{w.l.o.g, for $i\in [|1,m|]$, $k_{i,2}\neq 0$}}, else, just multiply $k_{x,2}$, $k_{y,2}$,\dots enough so that $M(k_{x,2}\delta_x+k_{y,2}\delta_y+k_{z,2}\delta_z+\dots)=\{y,z\}$ and set $k_{i,2}=1$).\newline\newline
We will now show that there exists $\lambda_1,\lambda_2\neq 0$ such that, with $W=2$,  $M(\lambda_1k_{x,1}\delta_x+\lambda_1k_{y,1}\delta_y+\lambda_1k_{z,1}\delta_z+\dots+\lambda_1k_{m,1}\delta_m+\lambda_2k_{x,2}\delta_x+\lambda_2k_{y,2}\delta_y+\lambda_2k_{z,2}\delta_z+\dots+\lambda_2k_{m,1}\delta_m)=\{y,z\}$

and with $W=1$, $M(\lambda_1k_{x,1}\delta_x+\lambda_1k_{y,1}\delta_y+\lambda_1k_{z,1}\delta_z+\dots+\lambda_1k_{m,1}\delta_m+\lambda_2k_{x,2}\delta_x+\lambda_2k_{y,2}\delta_y+\lambda_2k_{z,2}\delta_z+\dots+\lambda_2k_{m,1}\delta_m)=\{z\}$.\newline

To have that, if with $W=1$, $M(k_{x,2}\delta_x+k_{y,2}\delta_y+k_{z,2}\delta_z+\dots+k_{m,2}\delta_m)=\{z\}$, lets just take $\lambda_1=0,\lambda_2=1$.
\newline

Else, note $score_1(x),score_1(y),score_1(z)$ the scores of $x$, $y$ and $z$ on votes $k_{x,1}\delta_x+k_{y,1}\delta_y+k_{z,1}\delta_z+\dots+k_{m,1}\delta_m$ and $score_2(x),score_2(y),score_2(z)$ the scores of $x$, $y$ and $z$ on votes $k_{x,2}\delta_x+k_{y,2}\delta_y+k_{z,2}\delta_z+\dots+k_{m,2}\delta_m$.\newline

Take $\lambda_2$ such that $\lambda_2score_2(z)-\lambda_2score_2(x)>2*score_1(z)-score_1(x)$ ($score_2(z)-score_2(x)\neq 0$). 

Then, with $\lambda_1=\lceil \frac{\lambda_2score_2(y)-\lambda_2score_2(z)}{score_1(z)-score_1(y)}\rceil$ (or $\lambda_1=\frac{\lambda_2score_2(y)-\lambda_2score_2(z)}{score_1(z)-score_1(y)}+1$ if $\frac{\lambda_2score_2(y)-\lambda_2score_2(z)}{score_1(z)-score_1(y)}$ is an integer),

\textcolor{black}{\fbox{$\lambda_1score_1(z)+\lambda_2score_2(z)>\lambda_1score_1(y)+\lambda_2score_2(y)$ }}\newline
by definition of $\lambda_1$. \newline\newline
Now we will show that
\textcolor{black}{\fbox{$\lambda_1score_1(y)+\lambda_2score_2(y)>\lambda_1score_1(x)+\lambda_2score_2(x)$ }}\newline
  $\lambda_1$ is the smallest possible integer to get $\lambda_1score_1(z)+\lambda_2score_2(z)>\lambda_1score_1(y)+\lambda_2score_2(y)$ and \newline
$\lambda_2score_2(z)-\lambda_2score_2(x)>2*score_1(z)-*score_1(x)$ so \newline\newline
$\lambda_2score_2(y)+\lambda_1score_1(y)\ge$ \newline
$\lambda_2score_2(z)+\lambda_1score_1(z)-score_1(z)>$ ($\lambda_1\neq 0$ due to $M(k_{x,2}\delta_x+k_{y,2}\delta_y+k_{z,2}\delta_z)=\{y\}$)\newline
$\lambda_2score_2(x)+score_1(z)-score_1(x)+\lambda_1score_1(z)>$\newline
$\lambda_2score_2(x)+score_1(z)-score_1(x)+\lambda_1score_1(x)>$\newline
$\lambda_1score_1(x)+\lambda_2score_2(x)$\newline
Hence, there exists $k_{x,3},k_{y,3},k_{z,3},\dots,k_{m,3}=\lambda_1k_{x,1}+\lambda_2k_{x,2},\dots,\lambda_1k_{m,1}+\lambda_2k_{m,2}$ such that $M(k_{x,3}\delta_x+k_{y,3}\delta_y+k_{z,3}\delta_z)$ is a vector $\begin{pmatrix}
t_{x,3} \\
t_{y,3} \\
t_{z,3} \\
..
\end{pmatrix}$ with $t_{z,3}>t_{y,3}>t_{x,3}$, $t_{z,3}>t_{i,3}$ for every other object $i$ and $k_{x,3},k_{y,3},k_{z,3},\dots\neq 0$
\end{proof}

\begin{proposition*}{\bf~\ref{proof:4}}
For the total score functions, respecting CCP property implies respecting the set of constraints $\Delta$.
\end{proposition*}

\begin{proof}

Assume that $M$ does not verify $\Delta$. We already proved that if $M$ does not respect $\forall i,j,j\neq i~ M_{i,j}<M_{j,j}$, then it is not strategyproof, so it does not verify CCP, else, there exists $ x,y,z ~ s.t.~M_{y,z}-M_{y,x}\ge M_{z,z}-M_{z,x}$.\newline\newline
Lets take $ x,y,z ~ s.t.~M_{y,z}-M_{y,x}- M_{z,z}+M_{z,x}=\epsilon_1 >0$,\newline\newline
Due to Lemma \ref{lemma:a1}, there exist a vote $k_{x,1}\delta_x+k_{y,1}\delta_y+k_{z,1}\delta_z+\dots+k_{m,1}\delta_m$ such that $M(k_{x,1}\delta_x+k_{y,1}\delta_y+k_{z,1}\delta_z+k_{m,1}\delta_m) =\begin{pmatrix}
t_{x,1} \\
t_{y,1} \\
t_{z,1} \\
\dots
\end{pmatrix}$ with $t_{z,1}>t_{y,1}>t_{x,1}$, $t_{z,1}>t_{i,1}$ for every other object $i$ and $k_{x,1},k_{y,1},k_{z,1},\dots\neq 0$.

Lets set $\epsilon_2=M_{y,y}-M(z,y)$. Lets multiply $k_{x,1},k_{y,1},k_{z,1},\dots,t_{x,1},t_{y,1},t_{z,1},\dots$ by $\lceil \frac{\epsilon_2}{\epsilon_1} \rceil +1$ to obtain $k_{x,2},k_{y,2},k_{z,2},\dots,t_{x,2},t_{y,2},t_{z,2},\dots$. \newline Now, increase $k_{y,2}$ to $k_{y,3}$ and induce $k_{x,3},k_{z,3},\dots,t_{x,3},t_{y,3},t_{z,3},\dots$ so that $0<t_{z,3}-t_{y,3}\leq \epsilon_2$; it is possible because of the definition of $\epsilon_2$, and that keeps the order $t_{x,3}<t_{y,3}<t_{z,3}$, because $M_{y,y}>M_{x,z}$.\newline\newline
$k_{x,3}\ge \lceil \frac{\epsilon_2}{\epsilon_1} \rceil +1$, so we can change $k_{x,3}$ and $k_{z,3}$ to $k_{x,4}=k_{x,3}-\lceil \frac{\epsilon_2}{\epsilon_1} \rceil +1$ and $k_{z,4}=k_{z,3}+\lceil \frac{\epsilon_2}{\epsilon_1} \rceil +1$.\newline With $W=1$, $M(k_{x,3}\delta_x+k_{y,3}\delta_y+k_{z,3}\delta_z)_z>M(k_{x,3}\delta_x+k_{y,3}\delta_y+k_{z,3}\delta_z)_y$ but $M(k_{x,4}\delta_x+k_{y,3}\delta_y+k_{z,4}\delta_z)_y>M(k_{x,3}\delta_x+k_{y,3}\delta_y+k_{z,3}\delta_z)_z$. \newline Let W be the number of object that have a better score than $M(k_{x,3}\delta_x+k_{y,3}\delta_y+k_{z,3}\delta_z)_z$ plus one ($x$ and $y$ are not in the solution). By swapping the votes one by one, when $k_{x,3}$ gets changed to $k_{x,4}$ and $k_{z,3}$ to $k_{z,4}$, there is a swap from $x$ to $z$ that changes the solution. During the swap of a solution, $x$ cannot leave the solution since it was not in the solution, and $z$ cannot enter it, since it was already in the solution.  Thus, CCP is not verified.
\end{proof}

\begin{proposition*}{\bf~\ref{proposition:3:CCP}}
For the total score functions, with $3$ objects, respecting the set of equations $\Delta$ implies respecting property $CCP$.
\end{proposition*}

\begin{proof}

Let us suppose that
\begin{itemize}
    \item $M$ satisfies $\Delta$;
    \item A voter $\voter$ changes her ballot by removing $a$ and putting $b$ instead.
\end{itemize}
Let us denote $order_1$ (resp $order_2$) the order of the objects according to their score before (resp after) $\voter$ changes her vote.

Due to $ M_{c,b}-M_{c,a}<M_{b,b}-M_{b,a}$, if $c$ has a lower rank than $b$ in $order_1$, it still has a lower rank in $order_2$.\newline

Due to $M_{c,a}-M_{c,b}<M_{a,a}-M_{a,b}$, if $a$ has a higher rank than $c$ in $order_2$, it has a higher rank in $order_1$. Hence, if $c$ has a higher rank than $a$ in $order_1$, it still has a higher rank in $order_2$.\newline

Due to $M_{a,a}>M_{b,a}$ and $M_{b,b}>M_{a,b}$ (or with cases of equality and the tie break breaking for the object that is not voted), if $a$ has a higher rank than $b$ in $order_2$, it has a higher rank than $b$ in $order_1$, thus, if $b$ has a higher rank in $order_1$, it has a higher rank in $order_2$. 

The ties are broken the same way for $order_1$ and $order_2$.

With all these things coming together, the only  change that can happen in the winning set are the ones of CCP.
\end{proof}

\begin{proposition*}{\bf \ref{proof:6} }
Strategyproof total score functions respect $\Delta$. 
\end{proposition*}
\begin{proof}
Assume that $M$ does not verify $\Delta$. We already proved that if $M$ does not respect $\forall i,j,j\neq i~ M_{i,j}<M_{j,j}$, then it is not strategyproof, hence, there exists $ x,y,z ~ s.t.~M_{y,x}-M_{y,z}\ge M_{x,x}-M_{x,z}$.\newline\newline
Lets take $ x,y,z ~ s.t.~M_{y,x}-M_{y,z}- M_{x,x}+M_{x,z}=\epsilon_1 >0$,\newline\newline
Due to Lemma~\ref{lemma:a1}, there exists a set of ballots $\Declaration_1$, cast by a set of voters as large as necessary, such that $k_{x,1}\delta_x+k_{y,1}\delta_y+k_{z,1}\delta_z+\dots+k_{m,1}\delta_m$ and such that $M(k_{x,1}\delta_x+k_{y,1}\delta_y+k_{z,1}\delta_z+\dots+k_{m,1}\delta_m) =\begin{pmatrix}
t_{x,1} \\
t_{y,1} \\
t_{z,1} \\
\dots
\end{pmatrix}$ with $t_{z,1}>t_{y,1}>t_{x,1}$, $t_{z,1}>t_{i,1}$ for every other object $i$ and $k_{x,1},k_{y,1},k_{z,1},\dots\neq 0$.

There also exists a set of ballots $\Declaration_2$, cast by a set of voters as large as necessary, such that
$k_{x,2}\delta_x+k_{y,2}\delta_y+k_{z,2}\delta_z+\dots+k_{m,2}\delta_m$ and such that $M(k_{x,2}\delta_x+k_{y,2}\delta_y+k_{z,2}\delta_z+\dots+k_{m,2}\delta_m) =\begin{pmatrix}
t_{x,2} \\
t_{y,2} \\
t_{z,2} \\
\dots
\end{pmatrix}$, with $t_{z,2}>t_{x,2}>t_{y,2}$, $t_{z,2}>t_{i,2}$ for every other object $i$ and $k_{x,2},k_{z,2},k_{y,2},\dots\neq 0$
\newline

Let  $dist$ be $k_{x,2}-k_{y,2}$. For any $k_{\declaration_1}$, there exists $k_{\declaration_2}$ such that if \newline
$k_{x,3}=k_{\declaration_1}*k_{x,1}+k_{\declaration_2}*k_{x,2}$, $k_{y,3}=k_{\declaration_1}*k_{y,1}+k_{\declaration_2}**k_{y,2}$, $k_{z,3}=k_{\declaration_1}*k_{z,1}+k_{\declaration_2}*k_{z,2}$, etc\dots then \newline
$M(k_{x,3}\delta_x+k_{y,3}\delta_y+k_{z,3}\delta_z+k_{m,3}\delta_m) =\begin{pmatrix}
t_{x,3} \\
t_{y,3} \\
t_{z,3} \\
\dots
\end{pmatrix}$ with $t_{z,3}>t_{x,3}>t_{y,3}$, $t_z>t_i$ for every other object $i$, $k_{x,3},k_{z,3},k_{y,3},\dots\neq 0$, $t_{x,3}-t_{y,3}\le dist$ and due to the definition of $\Declaration_1$ and $\Declaration_2$, $t_{z,3}-t_{x,3}\ge k_{\declaration_1}*(t_{z,1}-t_{y,1})$. $t_{z,1}-t_{y,1}>0$ and $k_{\declaration_1}$ is as high as we want, so we will now assume that $t_{z,3}-t_{x,3}$ is the infinity and that the amount of votes for each object is also the infinity.
\newline

$dist$ is finite. Due to $M_{y,x}-M_{y,z}\ge M_{x,x}-M_{x,z}$, there exists a number $exch$ such that exchanging $exch$ votes from $z$ to $x$ leads to a better score for $y$ than for $x$. \newline

Let us multiply all the $t_{.,3}$ by $m!$ to obtain $t_{.,4}$. Now, note that any number between $1$ and $m$ divides the sum of votes for the objects, that is to say that for any $W$, with enough voters, such a set of ballots $\Declaration_4$ is possible to obtain, while keeping coherence between the number of objects in the ballot and the number of objects in the solution. Hence, it is now possible to choose $W$ as we want.\newline

Now denote by $order_1$ the list of objects ordered by the score given by $M(\Declaration_4)$. Now exchange $m! * exch$ votes from $z$ to $x$ (it is possible because there is an infinity of votes for $z$). This gives a new set of ballots $\Declaration_5$ in which we take $order_2$, the list of objects ordered by the score given by $M(\Declaration_5)$.\newline

Now let us set $W$ to be the rank of $y$ (in case of equality, which cannot happen between $x$ and $y$, by choice of $exch$, set $W$ so that every object in  equality with $y$ is selected) in $order_2$ and assume that a voter $\voter$ favorite outcome would be the one of $order_2$ (her strategyproof vote would be to vote for the objects selected in $order_2$). $W<m$ because $x$ is not selected. $z$ is selected both in $order_1$ and $order_2$, because it cannot have a lower score than $y$ in $order_2$ since it was at infinity in $order_4$ and we only changed a finite ($exch * m!$) amount of vote. The solution as changed between $order_1$ and $order_2$, because either $x$ was selected in $order_1$ and is not in $order_2$, either $x$ was not selected in $order_1$, hence, $y$ was not selected either, and $y$ is selected in $order_2$. Thus, the utility of $\voter$ is higher in $order_2$ than in $order_1$. There is one change in a ballot from $z$ to $x$ (from a set of ballots $\Declaration_5$ to a set of ballots $\Declaration_6$) that led to this improvement of utility. Assume that $\voter$ votes a strategyproof way and that the other ballots lead to $\Declaration_5$ (it is possible because there is an infinity of votes for every object). $z$ is part of the solution of $order_2$, which is optimal for $\voter$; hence, in her strategyproof vote, she votes for $z$ The same way, $x$ is not in her strategyproof vote. $\voter$ can now change her vote from $z$ to $x$ and increase her utility.
\end{proof}

\begin{proposition*}{\bf \ref{prop:936} }
With strictly more than $3$ objects, there exist some score functions that respect the set of constraints  $\Delta$ and are not strategyproof.
\end{proposition*}

\begin{proof} 
Let $M$ be a score matrix $
	\begin{pmatrix}
101 & 0 & 0 & 0\\
0 & 3 & 0 & 0 \\
0 & 0 & 100 & 0\\
0 & 2 & 0 & 99
\end{pmatrix}$. There are $2$ voters. The rows of the matrix correspond to objects $a,b,c,d$.

\begin{itemize}
    \item $\voter_1$ votes $a,d$
    \item $\voter_2$ votes $a,c$
\end{itemize}

$\{a,c\}$ wins. If, instead, the votes are:

\begin{itemize}
    \item $\voter_1$ votes $b,d$
    \item $\voter_2$ votes $a,c$
\end{itemize}

now $\{a,d\}$ wins. Hence, if $\voter_1$ wants $\{a,d\}$, she should lie on her ballot.

One can easily verify that this matrix respects $\Delta$
\end{proof}

\begin{proposition*}{\bf~\ref{proposition:p28}}
If $m>3$, every strategyproof total score function satisfies $\Delta^+$.
\end{proposition*}

\begin{proof}
Assume that $M$ does not verify $\Delta$. We already proved that if $M$ does not respect $\forall i,j,j\neq i~ M_{i,j}<M_{j,j}$ or $x,y,z ~ s.t.~M_{y,x}-M_{y,z}\le M_{x,x}-M_{x,z}$, then it is not strategyproof, hence, there exists $a,x,y,z \in \setObjects , M_{x,a}-M_{x,z}>M_{y,a}-M{y,z}$.\newline\newline
Lets take $ a,x,y,z\in \setObjects ~ s.t.~M_{x,a}-M_{x,z}- M_{y,a}+M{y,z}=\epsilon_1 >0$,\newline\newline
Due to Lemma \ref{lemma:a1}, there exists a set of ballots $\Declaration_1$ with as many voters as necessary such that $k_{x,1}\delta_x+k_{y,1}\delta_y+k_{z,1}\delta_z+\dots+k_{m,1}\delta_m$ and such that $M(k_{x,1}\delta_x+k_{y,1}\delta_y+k_{z,1}\delta_z+\dots+k_{m,1}\delta_m) =\begin{pmatrix}
t_{x,1} \\
t_{y,1} \\
t_{z,1} \\
\dots
\end{pmatrix}$ with $t_{z,1}>t_{y,1}>t_{x,1}$, $t_{z,1}>t_{i,1}$ for every other object $i$ and $k_{x,1},k_{y,1},k_{z,1},\dots\neq 0$

There also exists a set of ballots $\Declaration_2$ with as many voters as necessary such that
$k_{x,2}\delta_x+k_{y,2}\delta_y+k_{z,2}\delta_z+\dots+k_{m,2}\delta_m$ and such that $M(k_{x,2}\delta_x+k_{y,2}\delta_y+k_{z,2}\delta_z+\dots+k_{m,2}\delta_m) =\begin{pmatrix}
t_{x,2} \\
t_{y,2} \\
t_{z,2} \\
\dots
\end{pmatrix}$ with $t_{z,2}>t_{x,2}>t_{y,2}$, $t_{z,2}>t_{i,2}$ for every other object $i$ and $k_{x,2},k_{z,2},k_{y,2},\dots\neq 0$
\newline\newline
Let $dist$ be $k_{x,2}-k_{y,2}$. For any $k_{\declaration_1}$, there exists $k_{\declaration_2}$ such that if \newline
$k_{x,3}=k_{\declaration_1}*k_{x,1}+k_{\declaration_2}*k_{x,2}$, $k_{y,3}=k_{\declaration_1}*k_{y,1}+k_{\declaration_2}**k_{y,2}$, $k_{z,3}=k_{\declaration_1}*k_{z,1}+k_{\declaration_2}*k_{z,2}$, etc\dots then \newline
$M(k_{x,3}\delta_x+k_{y,3}\delta_y+k_{z,3}\delta_z+k_{m,3}\delta_m) =\begin{pmatrix}
t_{x,3} \\
t_{y,3} \\
t_{z,3} \\
\dots
\end{pmatrix}$ with $t_{z,3}>t_{x,3}>t_{y,3}$, $t_z>t_i$ for every other object $i$, $k_{x,3},k_{z,3},k_{y,3},\dots\neq 0$, $t_{x,3}-t_{y,3}\le dist$ and due to the definition of $\Declaration_1$ and $\Declaration_2$, $t_{z,3}-t_{x,3}\ge k_{\declaration_1}*(t_{z,1}-t_{y,1})$. $t_{z,1}-t_{y,1}>0$ and $k_{\declaration_1}$ is as high as we want, so we will now assume that $t_{z,3}-t_{x,3}$ is the infinity and that the amount of votes for each object is also the infinity.
\newline

$dist$ is finite. Due to $M_{y,x}-M_{y,z}\ge M_{x,x}-M_{x,z}$, there exists a number $exch$ such that exchanging $exch$ votes from $z$ to $a$ leads to a better score for $y$ than for $x$. \newline

Let us multiply all the $t_{.,3}$ by $m!$ to obtain $t_{.,4}$. Now, note that any number between $1$ and $m$ divides the sum of votes for the objects, that is to say that for any $W$, with enough voters, such a vote $\Declaration_4$ is possible to obtain, while keeping coherence between the number of objects in the ballot and the number of objects in the solution. Hence, it is now possible to choose $W$ as we want.\newline

Now denote by $order_1$ the list of objects ordered by the score given by $M(\Declaration_4)$. Now exchange $m! * exch$ votes from $z$ to $a$ (it is possible because there is an infinity of votes for $z$). This gives a new vote $\Declaration_5$ in which we take $order_2$, the list of objects ordered by the score given by $M(\Declaration_5)$.\newline

Now let us set $W$ to be the rank of $y$ (in case of equality, which cannot happen between $x$ and $y$, by choice of $exch$, set $W$ so that every object in equality with $y$ is selected) in $order_2$ and assume that a voter $\voter$ favorite outcome would be the one of $order_2$ (her strategyproof vote would be to vote for the objects selected in $order_2$). $W<m$ because $x$ is not selected. $z$ is selected both in $order_1$ and $order_2$, because it cannot have a lower score than $y$ in $order_2$ since it was at infinity in $order_4$ and we only changed a finite ($exch * m!$) amount of vote. The solution as changed between $order_1$ and $order_2$, because either $x$ was selected in $order_1$ and is not in $order_2$, either $x$ was not selected in $order_1$, hence, $y$ was not selected either, and $y$ is selected in $order_2$. Thus, the utility of $\voter$ is higher in $order_2$ than in $order_1$. There is one change a vote from $z$ to $x$ (from a vote $\Declaration_5$ to a vote $\Declaration_6$) that led to this improvement of utility. Assume that $\voter$ votes a strategyproof way and that the other ballots lead to $\Declaration_5$ (it is possible because there is an infinity of votes for every object). $z$ is part of the solution of $order_2$, which is optimal for $\voter$; hence, in her strategyproof vote, she votes for $z$ The same way, $x$ is not in her strategyproof vote. $\voter$ can now change her vote from $z$ to $x$ and increase her utility. 
\end{proof}

\begin{proposition*}{ \bf~\ref{proposition:p29} }
If $m>3$, every total score function that satisfies $\Delta^+$ satisfies CCP.
\end{proposition*}

\begin{proof}
Let us suppose that
\begin{itemize}
    \item $M$ satisfies $\Delta^+$;
    \item A voter $\voter$ changes her ballot by removing $a$ and putting $b$ instead.
\end{itemize}
Let us denote $order_1$ (resp $order_2$) the order of the objects according to their score before (resp after) $\voter$ changes her vote.

If $\forall c,d \in \setObjects , M_{d,b}-M_{d,a}=M_{c,b}-M{c,a}$, then every $c,d \in\setObjects$ have the same relative rank. \newline

Due to $b,c,a ~ ~M_{c,b}-M_{c,a}\le M_{b,b}-M_{b,a}$, for every object $c\neq a,b$, if $c$ has a lower rank than $b$ in $order_1$, it still has a lower rank in $order_2$.\newline

Due to $M_{a,b}<M_{b,b}$ and $M_{b,a}<M_{a,a}$ (or with cases of equality and the tie break breaking for the object that is not voted), if $a$ has a lower rank than $b$ in $order_1$, it still has a lower rank in $order_2$.\newline

Due to $b,c,a ~ ~M_{c,a}-M_{c,b}\le M_{a,a}-M_{a,b}$, for every object $c\neq a$, if $a$ has a higher rank than $c$ in $order_2$, it has a higher rank in $order_1$. Hence, if $c$ has a higher rank than $a$ in $order_1$, it still has a higher rank in $order_2$.\newline

The ties are broken the same way for $order_1$ and $order_2$.

With all these things coming together, the only  change that can happen in the winning set are the ones of CCP.
\end{proof}

\end{document}